\documentclass{tran-l}
\usepackage{graphics}
\usepackage{amssymb}

\makeatletter
\def\article@logo{}
\def\@setcopyright{}
\def\@adminfootnotes{%
  \let\@makefnmark\relax
  \let\@thefnmark\relax
  \ifx\@empty\thankses\else \@footnotetext{%
    \def\par{\let\par\@par}\@setthanks}%
  \fi
}
\makeatother

\newtheorem{thm}{Theorem}[section] 
\newtheorem{lem}[thm]{Lemma}     
\newtheorem{cor}[thm]{Corollary}
\newtheorem{prop}[thm]{Proposition}

\theoremstyle{definition}
\newtheorem{defn}[thm]{Definition}
\theoremstyle{remark}
\newtheorem{rem}[thm]{Remark}
\newtheorem{exa}[thm]{Example}
\numberwithin{equation}{subsection}

\newcommand{\eps}{\varepsilon}
\newcommand{\DE}{\Delta}
\newcommand{\de}{\delta}

\newcommand{\To}{\longrightarrow}

\newcommand{\U}{\mathcal{U}}
\newcommand{\V}{\mathcal{V}}
\newcommand{\CH}{\mathcal{H}}
\newcommand{\R}{\mathbb{R}}

\newcommand{\QQ}{\mathcal{Q}}
\newcommand{\Z}{\mathbb{Z}}
\newcommand{\N}{\mathbb{N}}

\newcommand{\FM}{\mathfrak{M}}

\newcommand{\m}{_{\! _{f\! B}} }
\newcommand{\hh}{_{\! _{f\! H}} } 
\newcommand{\hau}  {{_{\! H}} }
\newcommand{\boxx}{{_{\! B}} }

\begin{document}
\title[A Hausdorff dimension for finite sets]
 {A Hausdorff Dimension for Finite Sets}


\author{Juan M. Alonso}

\address{
   Juan M. Alonso\\
   Dpto. de Matem\'atica, Universidad Nacional de San Luis\\and UN de Cuyo\\ 
   Argentina}
\email{jmalonso@unsl.edu.ar}








\thanks{ 
\newline
\indent\emph{Acknowledgements.} The author would like thank the support of Secretar\'{i}a de Ciencia, T\'{e}cnica y Posgrado (SeCTyP) at UN Cuyo.}



\begin{abstract}
The classical Hausdorff dimension, denoted $\mbox{dim}_\hau $, of finite or countable sets is zero. We define an analog for finite sets, called \emph{finite Hausdorff dimension} and denoted $ \mbox{dim}\hh$,  which is non-trivial. It turns out that a finite bound for $\mbox{dim} \hh(F)$ guarantees that every point of $F$ has "nearby" neighbors. This property is important for many computer algorithms of great practical value, that obtain solutions by finding nearest neighbors. We also define $\mbox{dim} \m$, an analog for finite sets of the classical box-counting dimension, and compute examples. The main result of the paper is a Convergence Theorem. It gives conditions under which, if $F_{n}\rightarrow X$ (convergence of compact subsets of $\R^n$ under the Hausdorff metric), then $\mbox{dim} \hh(F_{n})\rightarrow \mbox{dim}_\hau (X)$.
\end{abstract}

\maketitle


\section{Introduction} 
\label{intro}
The initial motivation for this work was \emph{concentration of distance}. This is a particular instance of the \emph{curse of dimensionality}, a term coined by Richard Bellman in~\cite{Bellman}, to refer to various phenomena that arise in high-dimensional vector spaces. When searching for nearest neighbors, concentration of distance usually refers to the following: as the dimensionality of the data increases,
the longest and shortest distance between points tend to become so close that the distinction between "near" and "far" becomes meaningless. The lack of a clear contrast between distances to a query point compromises the quality of the search. The problem is a long standing one in database research \cite{Aggarwal,Beyer,Pramanik}. Awareness of this threat is spreading to other domains; in particular, major concerns have been raised in Cancer Research~\cite{Clarke}. This has prompted quite a bit of research aimed at better understanding both the problem and its implications~\cite{Durrant,Francois,Gianella,Hinneburg,Hsu,Kaban1,Kaban2,Pramanik,Radovanovich}.


In the papers cited above, concentration of distance is studied probabilistically. Data (a finite metric space embedded in a high-dimensional vector space) typically has some sort of structure that varies quite a lot depending on the data source (the different domains of application). We thought it worthwhile to try to understand this structure, however subtle it might be, in a more geometric vein. The ideas inherent in the study of fractals and fractal dimension seemed particularly appealing to us. In this regard~\cite{holst} was inspiring: the authors study "the form" of Word Space (one of the finite metric spaces mentioned above) using Statistics and a fractal dimension defined by considering "the underlying space as a continuum and randomly making a finite number of observations from which one tries to obtain a maximum likelihood approximation of the underlying dimension"\cite{holst2}. This is the usual way, when estimating dimension, of coping with finiteness~\cite{clarkson}.

In a radical departure from the classical theory, we decided to start directly with finite metric spaces. The problem, of course, is that the Hausdorff dimension of finite sets is zero.  In this paper we define \emph{finite Hausdorff dimension}, a non-trivial analog for finite spaces of the Hausdorff dimension. For the classical theory of Hausdorff and other, fractal, dimensions, see \cite{Falconer}, and the bibliography therein.

Throughout the paper we use $X,Y$, etc., to denote arbitrary metric spaces, and reserve $F,F'$, etc., to denote finite ones. Here is a summary of the contents of the paper. We begin Section~\ref{s:2} by recalling the definition of the classical Hausdorff measure and dimension. We then introduce 2-coverings. This key modification of the classical notion is responsible for making $\mbox{dim} \hh(F)$ non-zero on most finite $F$ (in fact, $\mbox{dim} \hh(F)=0$ if and only if $F$ has a single point). Two basic notions for this work, \emph{covering diameter} and \emph{focal points}, are also defined. The section ends with a brief discussion of how the results of the paper apply to the motivating problem: concentration of distance and the search for nearest neighbors.

Section~\ref{s:3} deals with the definition of finite Hausdorff dimension. Following Hausdorff's steps, we introduce $H^s(F)$, an analog of Hausdorff's outer measure $\CH^s(X)$. In contrast to the classical case, $H^s(F)$ is \emph{not} a measure. In Section~\ref{s:Holder} we study the behaviour of $H^s(F)$ under H\"{o}lder equivalences. The definition of $\mbox{dim} \hh(F)$ proper is given in Section~\ref{ss:DIMfH}. In the classical case, $\CH^s(X)$ has a "natural" break-point $s_0:=\mbox{dim}_\hau (X)$, with the property that $\CH^s(X)=0$, for all $s>s_0$, and $\CH^s(X)=\infty$, for all $s<s_0$. There is no such break-point in the finite case, hence the need to "manufacture" one. For this purpose, we consider the equation:
\begin{equation}\label{e:defDIMfH}
H^s(F)=\DE(F)^s,
\end{equation}
and solve for $s$, where $\DE(F)$ denotes the diameter of $F$. It turns out that (\ref{e:defDIMfH}) has a unique solution $s_0$ if and only if $F$ has no focal points. The solution is a positive real number, and we set $s_0:=\mbox{dim} \hh(F)$ (cf. Theorem~\ref{t:DIMf}).

In Section~\ref{s:4} we introduce, following the same pattern, \emph{finite box-counting dimension}, $\mbox{dim} \m(F)$. As in the classical case, there is an explicit formula to compute $\mbox{dim} \m(F)$.

In Section~\ref{s:5} we show that $\mbox{dim} \hh(F)\leq \mbox{dim} \m(F)$, just as in the classical case. We define \emph{locally uniform spaces}, and show this is a class of spaces where finite Hausdorff and finite box-counting dimensions coincide. Both finite dimensions are easier to compute for these spaces, and we even have an explicit formula.

Several examples are computed in Section~\ref{s:6}, together with results of a more general nature. For instance, we show that every non-negative extended real number is the finite Hausdorff [resp. finite box-counting] dimension of some finite space.

In Section~\ref{s:7} we prove the Convergence Theorems, the main results of the paper. Any compact space $X\subseteq \R^n$ can be approximated, in the Hausdorff metric, by a sequence $\{F_k\}$ of locally uniform spaces (Proposition~\ref{t:FapproxR^n}). In Section~\ref{ss:convThms} we prove, under some extra conditions on $X$ (cf. Theorem~\ref{t:ConvNet}), that
\[ \lim_{k\rightarrow \infty} \mbox{dim} \m(F_k)=\mbox{dim}_\boxx (X),\]
and when, moreover, $X$ is the attractor of an $IFS$ (cf. Theorem~\ref{t:ConvNetIFS}), that:
\[\lim_{k\rightarrow \infty} \mbox{dim} \hh(F_k)=\mbox{dim} _\hau (X).\]


\section{2-coverings and focal points.}\label{s:2}
For the benefit of the reader, we start by reviewing the classical definitions of Hausdorff measure and dimension (see~\cite{Falconer} for details). 

All subsets of $\R^n$ are metric spaces with the distance $d$ induced from $\R^n$. Recall that the \emph{diameter} of a non-empty subset $U$ of $\R^n$ is defined as $\DE(U):=\sup \{d(x,y)|x,y\in U\}$. Let  $\U=\{ U_i\}_{i=1}^\infty$ be a countable family of non-empty subsets of $\R^n$. We say that $\U$ has diameter at most $\de$, denoted $\DE(\U)\leq \de$, if $\DE(U_i)\leq \de$ for all $i$. The family $\U$ is called a covering of a subset $X$ of $\R^n$, if $X\subseteq \cup_{i=1}^\infty U_i$. For a covering $\U$ of $X$, and a number $s\geq 0$, we use the notation $\CH_\U^s(X):=\sum_{i=1}^\infty \DE(U_i)^s$.

Given a subset $X\subseteq \R^n$ and numbers $s\geq 0$ and $\de >0$, we define
\[
\CH _\de^s(X):=\inf\Big\{\CH_\U^s(X)\big|\U \textrm{ is a cover of }X,\textrm{ and }\DE(\U)\leq\de\Big\}.
\]
For fixed $s$, $\CH _\de^s(X)$ is clearly an increasing function of $\de$, hence the limit as $\de\to 0$ exists, and we define:
\[
\CH^s(X):=\lim_{\de\to 0} \CH _\de^s(X)=\sup_{\de>0} \CH _\de^s(X).
\]
Note that $\CH ^s(X)$ is defined for any subset $X$ of $\R^n$, and it is an extended number in $[0,\infty]$; it is called the \emph{s-dimensional Hausdorff (outer) measure} of $X$.

It turns out that there exists a critical value $s_0$ where $\CH ^s(X)$ jumps from $\infty$ to $0$. More precisely, for all $s>s_0$, $\CH ^s(X)=0$, and for all $s<s_0$, $\CH ^s(X)=\infty$. The \emph{Hausdorff dimension} of $X$ is defined to be this critical value:
\[
\dim_\hau(X):=s_0,
\]
and we have
\[
\CH ^s(X)=\left\{
\begin{array}{l l}
\infty &\quad\text{if $0\leq s<\dim_\hau(X)$}\\
0 &\quad \text{if $s>\dim_\hau(X)$.}
\end{array} \right.
\]

\subsection{Finite metric spaces}

Let $F$ denote a finite metric space. Unless explicit mention to the contrary, $F$ will be assumed to contain \emph{at least two elements}. We usually assume $F$ is contained in some metric space from which it inherits its metric. Although the finite dimensions are strongly dependent on the metric (cf. Example~\ref{e:metric}), we sometimes refer to $F$ as a set. The \emph{separation} of $F$, i.e. the minimum distance between different points of $F$, will be denoted $\de(F)$. Note that $0<\de(F)\leq \DE(F)<\infty$. We let $|F|$ denote the number of elements in $F$. The next definition is basic for this work.

\begin{defn}
A {\em $2$-covering} of $F$ is a family $\U=\{ U_i|i\in I\}$ of subsets of $F$ satisfying
\begin{enumerate}
 \item{$F=\bigcup_{i\in I} U_i$}
 \item{$|U_i|\geq 2, \forall i\in I$}
\end{enumerate}

\end{defn}

\begin{rem} In condition (ii) we depart from the classical definition. It is thanks to this condition that non-trivial dimensions can be assigned to finite spaces. Note that (ii) is equivalent to $U_i$ having \emph{positive} diameter. Finally, note that $I$ is finite, since $\U\subseteq \mathcal{P}(F)$, the power set of $F$.
\end{rem}

We denote the set of all $2$-coverings of $F$ by $K(F)$. There is exactly one $2$-covering which consists of one element, denoted $\U_0=\{F\}$. Notice that $\U_0\in K(F)$ because $|F|\geq 2$.

\begin{defn}
Let $\U\in K(F)$. The \emph{diameter} of $\U$, denoted $\Delta(\U)$, is defined by:
\[\DE(\U):=\textrm{ max}\{\DE(U_i)|U_i\in \U\}.\]
\end{defn}

\begin{defn}
Let $\U\in K(F)$. The \emph{covering diameter} of $F$, denoted $\nabla(F)$, is defined by:
\[\nabla(F):=\textrm{ min}\{\DE(\U)|\U\in K(F)\}.\]
\end{defn}

\begin{rem}
Note that $0<\de(F)\leq\nabla(F)\leq \DE(F)$.
\end{rem}

Given $\de>0$, we let $K_\de(F)$ denote the set of 2-coverings of $F$ with diameter $\le\de$:
\[K_\de(F):=\{\U\in K(F)|\DE(\U)\leq \de\}.\]

\begin{lem}\label{l:nabla}
Suppose $F$ is a finite set, and $\de>0$. The following conditions are equivalent:
\begin{enumerate}
\item{$K_\de(F)\neq \emptyset$,}
\item{$\de\geq \nabla(F).$}
\end{enumerate}
\end{lem}

\begin{proof}
Obvious.
\end{proof}

\begin{cor}
$\nabla(F)=\emph{ min}\{\de|K_\de(F)\neq \emptyset\}.$
\end{cor}

\begin{defn}
Let $\nu_F:F\rightarrow \R$, denote the function that gives the distance of a point to a nearest neighbor (in $F$). It is defined by
\[\nu_F(a):=\textrm{ min}\{d(a,x)|a,x\in F,\,\,x\neq a\}.\]
\end{defn}

\begin{lem}\label{l:2-cov}
Given a finite space $F$, suppose $r>0$ satisfies the condition that $\nu_F(a)\leq r$, for all $a\in F$. Then $K_r(F)$ is not empty.
\end{lem}

\begin{proof} Given $a\in F$, choose $x_a\in F$ such that $d(a,x_a)=\nu_F(a)$, and define $U_a:=\{a,x_a\}$. It follows that $\U:=\{U_a | a\in F\}$ is a 2-covering, and $\DE(U_a)\leq r$. In other words, $K_r(F)\neq\emptyset$, as desired.
\end{proof}

\begin{rem} \label{r:n-1elements}
We note, for further reference, that the 2-covering constructed in~\ref{l:2-cov} has $|F|$ elements. The reader can easily modify the construction and show that $F$ has a 2-covering with at most $|F|-1$ elements. In general this number is optimal: for example, $F:=\{(0,0),(1,0),(0,1),(-1,0),(0,-1)\}\subseteq \R^2$, cannot be 2-covered with less than 4 sets.
\end{rem}

\begin{prop}\label{p:delta,nabla} Let $F$ and $\nu_F$ be as above. Then:
\begin{enumerate}
\item{$\emph{ min}\,\{\nu_F(a)|a\in F\}=\,\de(F)$.}
\item{$\emph{ max}\{\nu_F(a)|a\in F\}=\nabla(F).$}
\end{enumerate}
\end{prop}

\begin{proof}
We simplify the notation by setting $m(F):=\textrm{min}\{\nu_F(a)|a\in F\}$, and $M(F):=\textrm{max}\{\nu_F(a)|a\in F\}$. To prove (i), notice that, $\nu_F(a)\ge\de(F)$, for all $a\in F$, so that $m(F)\ge\de(F)$. On the other hand, $\de(F)=d(a_0,a_1)$, for some $a_0,a_1\in F$. Then, $m(F)\leq \nu_F(a_0)\leq d(a_0,a_1)=\de(F)$, as desired.

To prove (ii), take $r=M(F)$ in Lemma~\ref{l:2-cov}, and let $\U$ be the 2-cover constructed in that lemma. By Lemma~\ref{l:nabla}, $M(F)\geq\nabla(F)$. To prove the reverse inequality, let $\nabla(F)=\DE(\V)$, for some 2-covering $\V$. Take an arbitrary $a\in F$, and suppose that $a\in V_i\in \V$. Then, for all $b\in V_i,\, b\neq a$, we have:
\[\nu_F(a)\leq d(a,b)\leq \DE(V_i)\leq\DE(\V)=\nabla(F).\]
It follows that $M(F)\leq \nabla(F)$, as desired. This concludes the proof.
\end{proof}

\subsection{Focal Points} In this section we introduce \emph{focal points}, an important notion for the rest of the paper. 
We start by decomposing $K(F)$ into three disjoint subsets
\[K(F)=K^0(F)\cup K^1(F)\cup K^2(F),\]
as follows: $K^0(F):=\{\U_0\}$,
\[
K^1(F):=\{\U\in K(F)|\,\,|\U|\geq 2\,\, \wedge \,\,\forall U_i\in\U,\, \DE(U_i) < \DE(F)\},
\]
and $K^2(F):=K(F)\setminus (K^0(F)\cup K^1(F))$. Thus, $\U\in K^2(F)$ if and only if $|\U|\geq2$, and $\exists U_i\in \U$, such that $\DE(U_i)=\DE(F)$. It is easy to see that:
\[K^ {1}(F):=\{\U\in K(F)|\DE(\U)<\DE(F)\}.\]
Finally, we set: $K^1_\de(F):=K^1(F)\cap K_\de(F)$.

\begin{defn}
Let $F$ be a finite space. We call $a_0\in F$ a {\em focal point} of $F$ if $\nu_F(a_0)=\DE(F)$. Explicitly, $d(a_0,a)=\DE(F)$, for all $a\in F, \,\, a\neq a_0$.
\end{defn}

\begin{rem}
In other words, \emph{all} neighbors of a focal point are "far away" (equally so, at diameter distance). A point is \emph{not} focal when it has "nearby" neighbors (i.e. neighbors at distances strictly less than $\DE(F)$). 
\end{rem}

The next result characterizes the existence of focal points in terms of $2$-coverings, covering diameter, and diameter. Note that condition (i) implies that $F$ has at least three points.

\begin{thm} \label{t:FOCAL}
Let $F$ be a finite space. Then the following conditions are equivalent:
\begin{enumerate}
\item{$F$ has no focal points,}
\item{$K^1(F) \neq \emptyset,$}
\item{$\nabla(F)<\DE(F)$.}
\end{enumerate}
\end{thm}

\begin{proof}
We assume that (i) holds and prove (iii). By definition, (i) means that $\nu_F(a)<\DE(F)$, for all $a\in F$. By Proposition~\ref{p:delta,nabla}, $\nabla(F)={{\max_{a\in F}}} \,\nu_F(a)<\DE(F)$, as desired.

We now show that (iii) implies (ii). Recall $K^1(F)$ can also be defined as the set of 2-covers $\U$ with $\DE(\U)<\DE(F)$. By Lemma~\ref{l:nabla}, $K_{\nabla(F)}(F)\neq \emptyset$, so that we can find a 2-cover $\U$ with $\DE(\U)=\nabla(F)<\DE(F)$. Hence $\U\in K^1(F)$, as required.

Finally, suppose (ii) holds. Let $\U\in K^1(F)$, and suppose, for contradiction, that $p\in F$ is a focal point. Let $p\in U_i\in\U$. For $a\in U_i$, $a\neq p$, we have \[\DE(F)=d(a,p)\leq\DE(U_i)\leq\DE(\U)<\DE(F),\]
a contradiction. This proves (i), and the Theorem.
\end{proof}

\subsection{Application to Nearest Neighbors.}\label{ss:nn}

Finding nearest neighbors in finite metric spaces is a method used to solve many important problems. The fields of application include Databases, Pattern Recognition, Computer Vision, DNA-Sequencing, Coding Theory, Data Compression, Text Analysis in real-time, Recommendation Systems, Spell Checking, Data Mining, etc.

Typically, one represents the objects of interest (and one's knowledge of them) by points in a vector space, and finds solutions by searching for a point nearest a given \emph{query point}. The whole concept is based on the assumption that nearby points have properties similar to those of the query point.

By the curse of dimensionality in the case of nearest neighbors, one usually means the phenomenon of \emph{concentration of distance}: the longest and shortest distance between points in the space are so close that the distinction between "near" and "far" becomes meaningless. In terms of the parameters we have introduced, this means that the quotient $\DE(F)/\delta(F)$ is close to one. Concentration of distance poses an obvious threat to solution methods based on finding nearest points. Hence the need to determine if the sets of points you usually obtain in your specific field of application, suffer from concentration of distance, and whether or not the problem is severe enough to defeat the assumption that "nearby" points have properties similar to those of the query point. 

In this section we discuss concentration of distance in light of the results obtained so far. Actually, we contend that rather than looking at the quotient $\DE(F)/\delta(F)$ or, equivalently, $\delta(F)/\DE(F)$, one should look at $\nabla(F)/\DE(F)$ instead. Indeed, we consider the question of \emph{how meaningful} a nearest neighbor is, and express the answer in term of this quotient. Our results are summarized in Theorem~\ref{t:omnibus} below (observe that this theorem includes notions and results obtained later in the paper). 
As usual, let $(F,d)$ denote an arbitrary finite metric space with at least two points.


\begin{defn}
Given arbitrary $x,x'\in F$, we say $x'$ is a point \emph{nearest} $x$, if $x\neq x'$, and $d(x,x'')\geq d(x,x')$, for all $x''\neq x$.
\end{defn}

\begin{rem}
The reader should be aware of the fact that we distinguish between "nearest \emph{point}", defined above, and "nearest \emph{neighbor}" to be defined presently. The difference lies with the notion of "neighbor" which, for us, excludes points lying "far away" (cf. Defs.~\ref{def:neigh} and \ref{def:nn}).
\end{rem}

\begin{lem}\label{l:nuNearest}
For any $x,x'\in F$, $x'$ is nearest $x$ iff $\nu(x)=d(x,x')$.
\end{lem}

\begin{defn}
An arbitrary function $n:F\rightarrow F$ is called a \emph{nearest point} function, if $n(x)$ is a point nearest $x$, for all $x\in F$.
\end{defn}

\begin{lem}
Every finite metric space has a nearest point function.
\end{lem}

\begin{rem}\label{r:fp} It follows that the existence of a nearest point function imposes no condition on $F$: such a function \emph{always} exists. This raises the question of \emph{meaningfulness} (cf. Def.~\ref{d:lambda-mean}).

Consider the definition of a nearest point function $n(x)$ at a focal point $x$. At $x$, we have exactly $|F|-1$ possible choices for $n(x)$, and no metric criterion to distinguish between them. So, \emph{any} such choice will give us a definition of a nearest point function but, clearly, distance gives no help to find points with properties similar to those of the query point.
\end{rem}

\begin{defn}\label{def:neigh}
Let $x,x'$ denote arbitrary points of $F$. We say that $x'$ is a \emph{neighbor} of $x$ if $x\neq x'$, and $d(x,x')<\DE(F)$.
\end{defn}

Intuitively, a neighbor of $x$ is a point different from $x$, and not far away from it. Clearly, a focal point has no neighbors. In fact:

\begin{lem}\label{l:noNgbrs=fp}
A point is focal iff it has no neighbors.
\end{lem}

\begin{defn}\label{def:nn}
An arbitrary function $n:F\rightarrow F$ is called a \emph{nearest neighbor} function, if $n(x)$ is a nearest neighbor of $x$, for all $x\in F$.
\end{defn}

\begin{lem}
A function $n:F\rightarrow F$ is nearest neighbor iff the following conditions hold:
\begin{enumerate}
\item{$d(x,n(x))=\nu_F(x)$, for all $x\in F$.}
\item{$\nu_F(x)<\DE(F)$, for all $x\in F$.}
\end{enumerate}
\end{lem}

Consider now the important question of when a nearest neighbor is meaningful. We believe this notion depends on the specific field of application: what is meaningful for databases need not be meaningful for, say, DNA-sequencing. This is why, instead of considering an absolute notion of meaningfulness, we introduce the following relative notion.

\begin{defn}\label{d:lambda-mean}
Let $\lambda$ denote a real number. An arbitrary function $n:F\rightarrow F$ is a \emph{$\lambda$-meaningful nearest neighbor} function, abbreviated $\lambda$-MNN function, if
\begin{enumerate}
\item{$n$ is a nearest neighbor function,}
\item{$0<\lambda<1$,}
\item{$\lambda$ is the smallest positive real number with $d(x,n(x))\leq \lambda\cdot \DE(F), \forall x\in F.$}
\end{enumerate}
\end{defn}

\begin{rem}
Note that there is always a point  $x_0\in F $, satisfying $d(x_0,n(x_0))=\nabla(F)$. It follows that, if  $n$ is a $\lambda$-MNN function, then $\lambda\geq \nabla(F)/\DE(F)$.
\end{rem}

\begin{lem}
Suppose $n:F\rightarrow F$ is a $\lambda$-MNN function. Then $\lambda=\nabla(F)/\DE(F)$.
\end{lem}

\begin{rem}
Note that $\lambda=\nabla(F)/\DE(F)$ is a sharp bound, since for some $x_0\in F$, 
\[
d(x_0,n(x_0))=\lambda\cdot\DE(F)=\nabla(F).
\]
\end{rem}

Taking advantage of results that will be proved in later sections, we can summarise the discussion in the following omnibus theorem:

\begin{thm} \label{t:omnibus}
Let $F$ be a finite metric space and $n:F\rightarrow F$ and arbitrary nearest point function. Then the following conditions are equivalent:
\begin{enumerate}
\item{$n$ is a nearest neighbor function.}
\item{$n$ is a $\nabla(F)/\DE(F)$-MNN function.}
\item{$n$ is a $\lambda$-MNN function.}
\item{$F$ has no focal points.}
\item{$K^1(F)$ is not empty.}
\item{$\nabla(F)/\DE(F)<1$.}
\item{$\emph{dim} \hh(F)$ is finite.}
\item{$\emph{dim} \m(F)$ is finite.}
\end{enumerate}
\end{thm}

\begin{rem}
(a) The equivalences (iv)-(vi) constitute Theorem~\ref{t:FOCAL}. The last two equivalences  follow from Theorems~\ref{t:DIMf} and \ref{t:DIMfB} below.
\end{rem}

It follows from the theorem that a nearest neighbor function is always $\lambda$-meaningful for a unique $\lambda\in (0,1)$, namely for $\lambda=\nabla(F)/\DE(F)$. Hence, both the existence of a nearest neighbor function, as well as its meaningfulness, depend on the quotient $\lambda=\nabla(F)/\DE(F)$: the function exists if $\lambda<1$, and it is more meaningful the smaller $\lambda$ is. The question for those working in a given field of application of nearest point search, then, is to decide whether or not $\lambda=\nabla(F)/\DE(F)$ is small enough so that, knowing that $d(x,n(x))\leq \lambda\cdot\DE(F)$, will guarantee that the "similarity" between $x$ and $n(x)$ is strong for their particular field.

We now consider our contention that concentration of distance is only \emph{partially} relevant to the nearest neighbor method. We begin by observing that \[0<\frac{\delta(F)}{\DE(F)}\leq \frac{\nabla(F)}{\DE(F)}\leq 1\,\cdot\]
It follows that, if $\de(F)/\DE(F)$ is close to 1, then so is $\nabla(F)/\DE(F)$ which, in turn, implies by Theorem~\ref{t:omnibus}, that we might not have a $\lambda$-MNN function. However, when $\de(F)/\DE(F)$ is small we cannot guarantee that also $\nabla(F)/\DE(F)$ will be small and, hence, we cannot be sure that we are in the clear. In fact, more is true: Example \ref{e:noCoD} below shows that $\de(F)/\DE(F)$ can be made as small as we please, while $\nabla(F)/\DE(F)=1$; thus $F$ has focal points and there is no nearest neighbor function. The example reveals that the quotient $\de(F)/\DE(F)$ fails to detect focal points.

\begin{exa}\label{e:noCoD}
Let $F:=\{A,B,C,D,E,G,O,H\}\subseteq \R^3$, where $A=(1,0,0),\,B=(1,0,1),\,C=(0,1,0),\,D=(0,1,1),\,E=(1,1,0),\,G=(1,1,1),\,
O=(0,0,0),\,H=(0,0,t)$, for $t\in(0,1]$. We consider $\R^3$ with the $\ell_\infty$-norm, and its associated distance $d=d_\infty$. Then all distances between different points are equal to 1, except for $d(O,H)=t$. It follows that $\de(F)=t$, and $\DE(F)=1=\nabla(F)$. Thus, while $\nabla(F)/\DE(F)=1$, and $A,B,C,D,E,G$ are focal points,  $\de(F)/\DE(F)=t$ can be made as small as we please. 
\end{exa}

\begin{rem}
It follows from Section~\ref{s:Holder} below, that the quotient:
\[\nabla(F)/\DE(F)=\left(\nabla(\eta(F))/\DE(\eta(F))\right)^\beta,\]
is preserved by similarities ($\beta=1$), but not by more general H\"{o}lder equivalences. For nearest neighbors it is important for the ratio to be small, so the above formula opens for the possibility of improving the ratio by transforming the space to a H\"{o}lder equivalent one. For instance, in Ex.~\ref{ex:fold} below, the inverse function $\eta':F'_n\rightarrow F_n$, i.e. the function that "unfolds" $F'_n$, is $(1,2)$-H\"{o}lder, and passing from $F'_n$ to $F_n$, the ratio is squared.
\end{rem}

\section{Finite Hausdorff dimension.}\label{s:3}

Following Hausdorff's steps, we start by introducing $H^s$, an analog of his outer measure $\CH^s$, and then use it to define the dimension proper. Later we relate $\mbox{dim} \hh(F)$ to the existence of {\em focal points} in $F$.

\subsection{Definition of $H^s$.} The functions $H^s_\U(F)$, $H^s_\de(F)$, and $H^s(F)$ defined in this section are analogs for finite spaces, of the classically defined functions $\CH^s_\U(X)$, $\CH^s_\de(X)$, and $\CH^s(X)$, respectively. In our context, these functions are interesting only when $F$ has no focal points, as will be seen when we define finite Hausdorff dimension later. 

\begin{defn}\label{d:H}
Let $F$ be a finite space with at least two elements, $s\in [0,\infty)$, and $\U\in K(F)$. Set
\[
H^s_\U(F):=\sum_{U_i\in \U}\DE(U_i)^s.
\]
Suppose, moreover, that $\de\geq \nabla(F)$. We then set
\[ 
H^s_\de(F) := \left\lbrace
  \begin{array}{l l}
    \textrm{min}\{H^s_\U(F)|\U\in K^1_\de(F)\}, & \emph{when $K^1(F)\neq\emptyset$},\\
    \textrm{min}\{H^s_\U(F)|\U\in K(F)\}, & \emph{when $K^1(F)=\emptyset$}.
  \end{array}
\right. 
\]
Finally,
\[H^s(F):={\textrm{max}}\{H^s_\de(F)|\de\geq\nabla(F)\}.\]
\end{defn}

\begin{lem}\label{l:Hdecreasing}
If $\nabla(F)\leq\de\leq \de'$, then $H^s_\de(F)\geq H^s_{\de'}(F)$.
\end{lem}
\begin{proof}
Obvious.
\end{proof}

\begin{lem}~\label{l:H}
For any finite space $F$ we have:
\[ H^s(F) = \left\lbrace
  \begin{array}{l l l}
    H^s_{\nabla(F)}(F)&=\emph{ min}\{H^s_\U(F)|\DE(\U)=\nabla(F)\}, & \textrm{when $K^1(F)\neq\emptyset$},\\
    \DE(F)^s&=\,\,H^s_\de(F), \emph{ for all } \de\geq\nabla(F), & \textrm{when $K^1(F)=\emptyset$}.
  \end{array}
\right. \]
\end{lem}

\begin{proof} By Lemma~\ref{l:Hdecreasing}, $H^s(F)=H^s_{\nabla(F)}(F)$. In case $K^1(F)\neq\emptyset$, we have $\U\in K^1_{\nabla(F)}(F)$ iff $\DE(\U)=\nabla(F)$. The result follows. When $K^1(F)=\emptyset$, we have seen in the proof of Lemma~\ref{l:Hdecreasing} that $K_{\nabla(F)}(F)=K(F)$. By definition, $H^s(F)=H^s_{\U_0}(F)=\DE(F)^s$, since $H^s_{\U_0}(F)<H^s_{\U}(F)$, for all $\U\in K^2(F)$. This completes the proof.
\end{proof}

Let $\U$ be a 2-cover of $F$, $\U=\{U_1,\dots,U_n\}$. Let $\{a_1,\cdots,a_k\}=\{\DE(U_i)|U_i\in\U\}$, $1\leq k\leq n$, denote the set of \emph{distinct} diameters of the elements of $\U$. We further assume that $a_1<a_2<\cdots <a_k$. Notice that $\de(F)\leq a_1$, and $\nabla(F)\leq\DE(\U)= a_k\leq\DE(F)$. With this notation:

\begin{equation}\label{e:H}
H^s_\U(F) = \left\lbrace
  \begin{array}{l l}
    m_1 a_1^s+m_2 a_2^s+\cdots+m_k a_k^s, & \textrm{when $\U\in K^1(F)\cup K^2(F)$},\\
    \DE(F)^s, & \textrm{when $\U=\U_0\in K^0(F)$},
  \end{array}
\right.
\end{equation}
where $ m_j\geq 1$, is the number of sets $U_i$ of diameter equal to $a_j$, so that $|\U|=m_1+\cdots+m_k$. In the first row of (\ref{e:H}), $a_k<\DE(F)$ when $\U\in K^1(F)$, and $a_k=\DE(F)$ when $\U\in K^2(F)$.

\begin{lem}\label{l:decrease}
Suppose that $\DE(F)=1$, and $K^1(F)\neq \emptyset$. Let $f(s)$ denote any of the following functions:
\begin{enumerate}
\item{$H^s_\U(F)$, for any $\U\in K^1(F)$,}
\item{$H^s_\de(F)$, for any $\de\geq \nabla(F)$,}
\item{$H^s(F)$.}
\end{enumerate}
Then $f$ is a positive, strictly decreasing function, $f(0)\geq 2$, and $\lim_{s\rightarrow \infty} f(s)=0$.
\end{lem}

\begin{proof} Consider (i). It follows directly from (\ref{e:H}) that $f(s)$ is positive, $f(0)=|\U|\geq 2$, and $f$ is strictly decreasing, because all $a_i<1$. The limit of $f(s)$ as $s$ goes to infinity is zero, because the same is true for every summand.

Consider now (ii). Since $K^1(F)\neq\emptyset,\,\,f(s)=\mbox{min}\{H^s_\U(F)|\U\in K^1_\de(F)\}$. Thus $f(s)$ is the minimum of a finite number of functions that satisfy all required conditions, by (i). Hence so does $f(s)$, as desired. Finally, (iii) is the special case of (ii) when $\de=\nabla(F)$, by Lemma~\ref{l:H}. This completes the proof.
\end{proof}

\subsection{H\"{o}lder equivalences.} \label{s:Holder}
In this section we study the behavior of $H^s(-)$ with respect to H\"{o}lder \emph{equivalences}, and note that the usual relaxations (H\"{o}lder condition, Lipschitz and bi-Lipschitz condition) impose essentially no condition on finite metric spaces.

\begin{defn} A function $\eta:(X,d)\rightarrow (X',d')$ is called a \emph{H\"{o}lder equivalence} if there is $r,\beta>0$ such that
\[d'(\eta(x_1),\eta(x_2))=r d(x_1,x_2)^\beta\]
for all $x_1,x_2\in X$. We say that $\eta$ is \emph{$(r,\beta)$-H\"{o}lder}, or an \emph{$(r,\beta)$-H\"{o}lder equivalence}. In the special case when $\beta=1$, we say that $\eta$ is a \emph{similarity}, or an \emph{r}-similarity.
\end{defn}

\begin{exa}\label{ex:fold}
This example is obtained by "folding" an equally spaced linear set. Let $F_n:=\{x_0,\dots,x_{n-1}\}\subseteq \R$ consist of the following $n$ points: $x_i=i,\, (i=0,\dots,n-1)$. Then $d(x_i,x_{i+j})=|j|$. Consider the space $F_n':=\{y_0,\dots,y_{n-1}\}\subseteq \R^n$, where
\[y_i:=(\overbrace{1,\dots,1}^i,0,\dots,0) .\]
Then for $i,j\geq 0$,
\[d'(y_i,y_{i+j})=\|(\overbrace{0,\dots,0}^i,\overbrace{1,\dots,1}^j,0,\dots,0)\|= \sqrt{j}.\]
Define $\eta:F_n\rightarrow F_n'$, by $\eta(x_i):=y_i$. Then $d'(\eta(x_i),\eta(x_{i+j}))=\sqrt{j}$, and $d(x_i,x_{i+j})=j$. In other words, $\eta$ is $(1,1/2)$-H\"{o}lder.
\end{exa}

\begin{lem}\label{l:(r,a)-H}
Suppose that $\eta:X\rightarrow X'$ is $(r,\beta)$-H\"{o}lder. Then:
\begin{enumerate}
\item{For all $Y\subseteq X$, $\eta:Y\rightarrow \eta(Y)$ is a bijection, and its inverse $\eta'$ is $(r^{-1/\beta},1/\beta)$-H\"{o}lder.}
\item{If $F\subseteq X$ is finite, then $\DE(\eta(F))=r\DE(F)^\beta$.}
\item{If $F\subseteq X$ is finite and has focal points, then so does $\eta(F)\subseteq X'$.}
\item{Let $F\subseteq X$ be finite. Then $F$ has focal points iff $\eta(F)$ has focal points.}
\end{enumerate}
\end{lem}

\begin{proof}
(i) It is obvious from the definition that $\eta$ is injective. Let $\eta':\eta(Y)\rightarrow Y$ be the inverse of $\eta$, and let $x',y'\in \eta(Y)$. Then $x'=\eta(x),\,y'=\eta(y)$ for unique $x,\,y\in Y$, and $d'(\eta'(x'),\eta'(y'))=d(x,y)=[(1/r)d'(\eta(x),\eta(y))]^{1/\beta}$. (i) follows immediately.

(ii) Suppose $\DE(\eta(F))=d'(\eta(x_1),\eta(x_2))$, for $x_1,x_2\in F$. Then $\DE(\eta(F))=rd(x_1,x_2)^\beta\leq r\DE(F)^\beta$. For the reverse inequality, assume  $\DE(F)= d(u_1,u_2),\,u_i\in F$. Then $rd(u_1,u_2)^\beta=d'( \eta(u_1),\eta(u_2))\leq \DE(\eta(F))$, as desired.

(iii) Suppose $x_0\in F$ is focal. Then $d(x_0,x)=\DE(F)$, for all $x\neq x_0$. Hence $d'(\eta(x_0),\eta(x))=rd(x_0,x)^\beta=r\DE(F)^\beta=\DE(\eta(F))$, by (ii). This shows that $\eta(x_0)$ is a focal point of $\eta(F)$, because every $x'\in \eta(F)$ different from $\eta(x_0)$, is of the form $\eta(x)$ for some $x\in F$, $x\neq x_0$.

Finally, (iv) follows immediately from (i) and (iii).
\end{proof}

\begin{lem}\label{l:simi}
Let $\eta:X\rightarrow X'$ be $(r,\beta)$-H\"{o}lder, and $F\subseteq X$ a finite set. Then:
\begin{enumerate}
\item{$\eta$ induces bijections:
\begin{enumerate}
\item{$\eta_*:K(F)\rightarrow K(\eta(F))$,}
\item{$\eta_*:K_\de(F)\rightarrow K_{r\de^\beta}(\eta(F))$, for all $\de\geq\nabla(F)$,}
\item{$\eta_*:K^1_\de(F)\rightarrow K^1_{r\de^\beta}(\eta(F))$, for all $\de\geq\nabla(F)$.}
\end{enumerate}}
\item{$|\eta_*(\U)|=|\U|$.}
\item{$\DE(\eta_*(\U))=r\,\DE(\U)^\beta$.}
\item{$\nabla(\eta(F))=r\,\nabla(F)^\beta$.}
\end{enumerate}
\end{lem}

\begin{proof}
(i) For $\U=\{U_1,\dots,U_n\}\in K(F)$, define $\eta_*(\U):=\{\eta(U_1),\dots,\eta(U_n)\}$. Then $\eta_*(\U)$ is a 2-covering because $\eta|F:F\rightarrow \eta(F)$ is bijective, with inverse $\eta'$. Indeed, if $F=\bigcup U_i$, then $\eta(F)=\bigcup \eta(U_i)$, and $|\eta(U_i)|=|U_i|\geq 2$, as required. To see that $\eta_*$ is bijective, recall that $\eta'$ is $(r^{-1/\beta},1/\beta)$-H\"{o}lder, by Lemma~\ref{l:(r,a)-H}(i). We then have $\eta'_*:K_{\de'}(\eta(F)) \rightarrow K_{(\de'/r)^{1/\beta}}(F)$ and, clearly, $\eta_*$ and $\eta'_*$ are inverse to each other. This completes the proof of (a). Using Lemma~\ref{l:(r,a)-H}(ii), if $\U\in K_\de(F)$, then  $\DE(\eta(U_i))=r\DE(U_i)^\beta\leq r\de^\beta$. Hence $\eta_*(\U)\in K_{r \de^\beta}(F)$, as desired. Similarly, $\eta_*(K^1(F))\subseteq K^1(\eta(F))$, because $\DE(\U)<\DE(F)$ implies $\DE(\eta(\U))=r\DE(\U)^\beta<r\DE(F)^\beta=\DE(\eta(F))$.

(ii) is obvious from the definition of $\eta_*$, and (iii) follow immediately from (i) and Lemma~\ref{l:(r,a)-H}. To prove (iv), we use (i) and (iii):

\begin{equation}\nonumber
\begin{array}{l l}
 \nabla(\eta(F)) &=  \textrm{min }\{\DE(\V)|\V\in K(\eta(F))\}\\
    &=\textrm{min }\{\DE(\eta_*(\U))|\U\in K(F)\}\\
    &=\textrm{min }\{r(\DE(\U))^\beta|\U\in K(F)\}\\
    &=r\nabla(F)^\beta.
  \end{array}
\end{equation}
This proves (iv), and concludes the proof of the lemma.
\end{proof}

\begin{prop}\label{p:H-simi}
Let $\eta:X\rightarrow X'$ be $(r,\beta)$-H\"{o}lder, and $F\subseteq X$ a finite space. Then, for all $s\in[0,\infty)$:
\begin{enumerate}
\item{$H^s_{\eta_*(\U)}(\eta(F)) = r^s H^{s\beta} \U(F)$, for all $\U\in K(F)$.}
\item{$H^s_{r\de^\beta}(\eta(F)) = r^s H^{s\beta} \de(F)$, for all $\de\geq \nabla(F)$.}
\item{$H^s(\eta(F)) = r^s H^{s\beta}(F)$.}
\end{enumerate}
\end{prop}

\begin{proof}
(i) Let $\U=\{U_i\}\in K(F)$. By Lemma~\ref{l:simi}, $\eta_*(\U)\in K(\eta(F))$, and:
\[H^s_{\eta_*(\U)}(\eta(F))=\sum \DE(\eta(U_i))^s= \sum r^s\DE(U_i)^{s\beta}=r^sH^{s\beta} \U(F).\]
(ii) Given $s,\de$, suppose $H^{s\beta} \de(F)=H^{s\beta} \U(F)$, where (a) $\U\in K^1_\de(F)$, when $F$ has no focal points, and (b) $\U\in K(F)$, otherwise. We consider (a) first. By (i) and Lemma~\ref{l:simi}, $
H^s_{\eta_*(\U)}(\eta(F))=r^sH^{s\beta} \de(F)$. Since $\eta_*(\U)\in
K^1_{r\de^\beta}(F)$, $H^s_{r\de^\beta}(\eta(F))\leq H^s_{\eta_*(\U)}(\eta(F))=r^s H^{s\beta} \de(F)$. For the reverse inequality, let $H^s_{r\de^\beta}(\eta(F))=H^s_\V(\eta(F))$, for some $\V\in K^1_{r\de^\beta}(\eta(F))$. Using Lemma~\ref{l:simi}, $\eta'_*(\V)\in K^1_\de(F)$, and
\[
H^{s\beta} \de(F)\leq H^{s\beta} {\eta'_*(\V)}(F)=\frac  {1}{r^s}H^s_\V(\eta(F))=
\frac  {1}{r^s}H^s_{r\de^\beta}(\eta(F)).
\]
This completes the proof of (ii) in case (a). The proof in case (b) is similar: we need only use the fact that now $\eta_*(\U)\in K(\eta(F))$ and, if $\V\in K(\eta(F))$, then $\eta'_*(\V)\in K(F)$.

(iii) is a special case of (ii). Here are the details. Suppose first that $F$ has no focal points. By Lemma~\ref{l:(r,a)-H}(iv), the same is true of $\eta(F)$. According to Lemma ~\ref{l:H}, $H^{s\beta}(F)=H^{s\beta} {\nabla(F)}(F)$, and $H^s(\eta(F))=H^s_{\nabla(\eta(F))}(\eta(F))$. By Lemma~\ref{l:simi}(iv), $\nabla(\eta(F))=r\nabla(F)^\beta$. Using (ii),
\[
H^{s\beta} {\nabla(F)}(F)=\frac  {1}{r^s}H^s_{r(\nabla(\eta(F)))^\beta}(\eta(F)) =\frac  {1}{r^s}H^s_{\nabla(\eta(F))}(\eta(F))
\]
Hence, $r^sH^{s\beta}(F)=H^s(\eta(F))$, as desired. The case when $F$ has focal points will be left to the reader. This completes the proof.
\end{proof}

Recall the following relaxations of H\"{o}lder equivalence and of similarity, defined here for arbitrary metric spaces.

\begin{defn}\label{d:holder}
Let $X,X'$ be metric spaces, $\eta:X\rightarrow X'$ a function, and $r,\beta>0$. Then:
\begin{enumerate}
\item{$\eta$ satisfies a \emph{H\"{o}lder condition}, or an $(r,\beta)$\emph{-H\"{o}lder condition}, if \[d'(\eta(x),\eta(y))\leq rd(x,y)^\beta.\]}
\item{$\eta$ is \emph{Lipschitz} if it satisfies an $(r,1)$-H\"{o}lder condition, for some $r>0$.}
\item{$\eta$ is \emph{bi-Lipschitz} if \[r_1d(x,y)\leq d'(\eta(x),\eta(y))\leq r_2d(x,y),\] for some $r_1,r_2>0$. We say that $X$ and $\eta(X)$ are \emph{Lipschitz equivalent}.}
\end{enumerate}
\end{defn}

It turns out that these relaxations are not so interesting in the finite case as they are in the classical case. This is shown by the following lemma, whose easy proof we leave to the reader.

\begin{lem}\label{l:simNOcond}
Suppose $\eta:F\rightarrow X'$ is a function defined on a finite $F\subseteq X$. Then
\begin{enumerate}
\item{Any such $\eta$ is Lipschitz.}
\item{$\eta$ is bi-Lipschitz iff it is injective.}
\item{$F$ and $\eta(F)$ are Lipschitz equivalent iff $|F|=|\eta(F)|$.}
\end{enumerate}
\end{lem}


\subsection{Definition of $\mbox{dim} \hh(F)$.} \label{ss:DIMfH}
We define \emph{finite Hausdorff dimension}, $\mbox{dim} \hh(F)$, by solving the equation:
\begin{equation}\label{e:defDIM}
H^s(F)=\DE(F)^s.
\end{equation}
Equation (\ref{e:defDIM}) has exactly one solution $s_0\in(0,\infty)$, precisely when $F$ has no focal points. More generally, we have:

\begin{prop}\label{p:uniqueSol}
Consider the following equations:
\begin{enumerate}
\item{$\DE(F)^s=H^s_\U(F)$, for all $\U\in K(F)$,}
\item{$\DE(F)^s=H^s_\de(F)$, for all $\de\geq \nabla(F)$,}
\item{$\DE(F)^s=H^s(F)$.}
\end{enumerate}
Then, in each of these cases, the equation has a unique solution iff $F$ has no focal points. When this is the case, the solutions are positive real numbers, and will be denoted, respectively, $s_\U,\,s_\de$, and $s_0$.
\end{prop}

\begin{proof}
Suppose $F$ is a subspace of $(X,d)$. The identity map:
\[
\textrm{id}_X:(X,d)\rightarrow (X,\frac  {1}{r} d)
\]
is an $r^{-1}$-similarity. To prove (i), note that by Proposition ~\ref{p:H-simi} (i), $H^s_{(\textrm{id}_X)_*(\U)}(\textrm{id}_X(F)) = r^{-s} H^s_\U(F)$. Taking $r=\DE(F)$, we see that $\DE(F)^s=H^s_\U(F)$ is equivalent to:
\begin{equation}~\label{e:equiv}
\DE(\textrm{id}_X(F))^s=1=\frac  {1}{\DE(F)^s}H^s_\U(F)=H^s_{(\textrm{id}_X)_*(\U)} (\textrm{id}_X(F)).
\end{equation}
It follows from (\ref{e:equiv}) that in the proof of (i) we may assume, without loss of generality, that $\DE(F)=1$. Consider first the reverse implication. If $F$ has no focal points, Lemma~\ref{l:decrease} guarantees the existence of a unique $s_\U\in (0,\infty)$ such that $H^{s_\U} \U(F)=1$, as desired. To prove the direct implication, suppose that $K^1(F)= \emptyset$. Using  Def.~\ref{d:H} we see that the equation in (i) has infinitely many solutions when $\U\in K^0(F)$, and no solution when $\U\in K^2(F)$. This completes the proof of (i).

The proof of (ii) is completely analogous, except, perhaps, for the last part. So assume $K^1(F)=\emptyset$. Then $H^s_\de(F)=H^s(F)=\DE(F)^s$, by Lemma ~\ref{l:H}, and the equation has infinitely many solutions, as before. Finally, (iii) is a special case of (ii). This completes the proof.
\end{proof}

\begin{lem}
Suppose $F$ has no focal points. Then
\begin{enumerate}
\item{$s_0=s_{\nabla(F)}$.}
\item{$s_0=\emph{max}\{s_\de|\de\geq \nabla(F)\}$.}
\item{$s_\de=\emph{min}\{s_\U|\U\in K^1_\de(F)\}$, for all $\de\geq \nabla(F)$.}
\end{enumerate}
\end{lem}

\begin{proof}
(i) is obvious, since $H^s(F)=H^s_{\nabla(F)}(F)$. (ii) If $\de\geq \nabla(F)$, then $H^s_\de(F)\leq H^s_{\nabla(F)}(F)$ by Lemma~\ref{l:Hdecreasing}. Hence $s_\de\leq s_{\nabla(F)}=s_0$. Thus, $\textrm{max}\{s_\de|\de\geq\nabla(F)\}=s_0$, as required. To prove (iii), note that Def.~\ref{d:H} implies $H^s_\de(F)\leq H^s_\U(F)$, for all $\U\in K^1_\de(F)$; hence, $s_\de\leq \textrm{min}\{s_\U|\U\in K^1_\de(F)\}$. To prove the reverse inequality, recall that, given $s_\de$ there is $\U\in K^1_\de(F)$ such that $\DE(F)^{s_\de}=H^{s_\de} \de(F)=H^{s_\de} \U(F)$. By Proposition ~\ref{p:uniqueSol}, $s_\de=s_\U$. This completes the proof.
\end{proof}

\begin{defn}\label{d:DIMf}
For a finite, non-empty subset $F\subseteq \mathbb{R}^n$, we define
\[ \textrm{dim} \hh(F) := \left\lbrace
  \begin{array}{c l}
    0 & \textrm{if $|F|=1$},\\
    \infty & \textrm{if } K^1(F)=\emptyset,\\
    s_0 & \textrm{if } K^1(F)\neq\emptyset
    .
  \end{array}
\right. \]
\end{defn}
We can summarize our results so far as follows
\begin{thm}\label{t:DIMf}
Let $F$ be a non-empty, finite set. Then $\emph{dim} \hh(F)$ is a positive real number if and only if $F$ has no focal points; it is infinity if and only if $F$ has focal points, and it is zero when $F$ has one element.
\end{thm}

\begin{thm}\label{t:dimHolder}
Let $\eta:X\rightarrow X'$ be $(r,\beta)$-H\"{o}lder, and $F\subseteq X$ a finite space. Then
\begin{equation}\nonumber
\beta\cdot\emph{dim} \hh(\eta(F))=\emph{dim} \hh(F).
\end{equation}
In particular, $\emph{dim} \hh$ is preserved by similarities.
\end{thm}

\begin{proof}
Suppose first that $F$ has no focal points. Then, $\mbox{dim} \hh(\eta(F))$ is the unique solution of the equation
\begin{equation}\label{e:dimSim}
H^s(\eta(F))=\DE(\eta(F))^s.
\end{equation}
Using Proposition ~\ref{p:H-simi} and Lemma ~\ref{l:simi}, we see that (\ref{e:dimSim}) is equivalent to $H^{s\beta}(F)=\DE(F)^{s\beta}$, whose only solution is $\mbox{dim} \hh(F)$, as desired. By Lemma~\ref{l:(r,a)-H}, $F$ has focal points [resp. $|F|=1$] if and only if $\eta(F)$ has focal points [resp. $|\eta(F)|=1$]. Hence, the dimension of $F$ is infinity [resp. zero] if and only if the dimension of $\eta(F)$ is infinity [resp. zero]. This completes the proof.
\end{proof}

\begin{thm} \label{t:boundsS_U}
Let $F$ be a finite space with no focal points. Suppose $\U\in K^1_\de(F)$, for some $\de\geq \nabla(F)$, and let $a_1$ [resp. $a_k$] denote the smallest [resp. largest] diameter of elements of $\U$. Then
\begin{enumerate}
\item{\begin{equation}\label{e:boundsU}
\frac{\ln|\U|}{\ln\frac{\DE(F)}{\de(F)}} \leq \frac{\ln|\U|}{\ln\frac{\DE(F)}{a_1}} \leq s_\U\leq \frac{\ln|\U|}{\ln\frac{\DE(F)}{a_k}} \leq \frac{\ln|\U|}{\ln\frac{\DE(F)}{\de}}\end{equation}}
\item{\begin{equation}\label{e:boundsU2}
\frac{\DE(F)}{\de} \leq \frac{\DE(F)}{a_k} \leq |\U|^{1/s_\U}\leq \frac{\DE(F)}{a_1} \leq \frac{\DE(F)}{\de(F)}\end{equation}}
\end{enumerate}
\end{thm}

\begin{proof}
By equation (\ref{e:H}), $H^s_\U(F)=m_1 a_1^s+\cdots+m_k a_k^s$. Hence,
\begin{equation}\label{e:f_i}
|\U|\,\de(F)^s\leq |\U|\, a_1^s\leq H^s_\U(F) \leq |\U|\,a_k^s\leq|\U|\, \de^s
\end{equation}
We introduce the following definitions, as shorthand: $f_1(s,\U):=|\U|\,\de(F)^s$, $f_2(s,\U):=|\U|\, a_1^s$, $f_3(s,\U):=H^s_\U(F)$, $f_4(s,\U):=|\U|\,a_k^s$, $f_5(s,\U):=|\U|\, \de^s$. If $\eta$ is an $r$-similarity, $f_i(s,\eta_*(\U))=r^sf_i(s,\U)$, by the results of Section~\ref{s:Holder}. It follows that the equations
\begin{equation}\label{e:f_i=Delta}
f_i(s,\eta_*(\U))=\DE(\eta(F))^s,\quad\mbox{ and }\quad f_i(s,\U)=\DE(F)^s,
\end{equation}
are equivalent, i.e. have the same solutions. As in the proof of Proposition~\ref{p:uniqueSol}, we may assume, without loss of generality, that $\DE(F)=1$. When this is the case, all five functions $f_i$ are decreasing, $f_i(0,\U)=|\U|$, and they all tend to zero when $s$ goes to infinity. Since every number in (\ref{e:boundsU}) is the solution of an equation (\ref{e:f_i=Delta}), and these can be computed solving $f_i(s,\U)=1$, we see that (\ref{e:boundsU}) follows from (\ref{e:f_i}).

(ii) (\ref{e:boundsU2}) is an immediate consequence of (\ref{e:boundsU}). This proves the theorem.
\end{proof}

\begin{cor}\label{c:boundsU}
Suppose $\de=\nabla(F)$, and $\U\in K^1_{\nabla(F)}(F)$. Then
\begin{enumerate}
\item{\[\frac{\ln|\U|}{\ln\frac{\DE(F)}{\de(F)}} \leq s_\U\leq \frac{\ln|\U|}{\ln\frac{\DE(F)}{\nabla(F)}} \]}
\item{\[\frac{\DE(F)}{\nabla(F)} \leq |\U|^{1/s_\U}\leq \frac{\DE(F)}{\de(F)}\cdot\]}
\end{enumerate}
\end{cor}

The first upper bound in the next corollary follows from Remark~\ref{r:n-1elements}.
\begin{cor}
Suppose $F$ has no focal points. Then
\begin{enumerate}
\item{\[\frac{\ln 2}{\ln\frac{\DE(F)}{\de(F)}} \leq  \emph{dim} \hh(F) \leq \frac{\ln(|F|-1)}{\ln\frac{\DE(F)}{\nabla(F)}}\]}
\item{\[2\leq \big[\DE(F)/\de(F)\big]^{\emph{dim} \hh(F)},\quad \text{and}\quad \big[\DE(F)/\nabla(F)\big]^{\emph{dim} \hh(F)} \leq |F|-1. \]}
\end{enumerate}
\end{cor}

\section{Finite Box Dimension, $\mbox{dim} \m(F)$.} \label{s:4}
The classical box-counting (or Minkowski-Bouligand) dimension will be denoted $\mbox{dim}_\boxx (-)$. In this section we define an analog for finite metric spaces, denoted $\mbox{dim} \m(-)$, and called \emph{finite box dimension}. We follow the same pattern we used to define finite Hausdorff dimension. The proofs for finite box-dimension are similar, but usually simpler, than those for finite Hausdorff dimension, and will be left to the reader.

\begin{defn}  \label{d:B}
For $\U\in K(F)$, set
\[B^s_\U(F):=|\U|\,\DE(\U)^s.\]
For $\de\geq \nabla(F)$, set:
\[ B^s_\de(F) := \left\lbrace
  \begin{array}{l l}
    \textrm{min}\{B^s_\U(F)|\U\in K^1_\de(F)\}, & \emph{when $K^1(F)\neq\emptyset$},\\
    \textrm{min}\{B^s_\U(F)|\U\in K(F)\}, & \emph{when $K^1(F)=\emptyset$}.
  \end{array}
\right. \]
Finally,
\[B^s(F):={\textrm{max}}\{B^s_\de(F)|\de\geq\nabla(F)\}.\]
\end{defn}

\begin{lem}\label{l:Bdecreasing}
Suppose $\nabla(F)\leq\de\leq \de'$. Then $B^s_\de(F)\geq B^s_{\de'}(F)$.
\end{lem}

\begin{lem}~\label{l:B}
For any finite space $F$ we have:
\[ B^s(F) = \left\lbrace
  \begin{array}{l l l}
    B^s_{\nabla(F)}(F)&=\emph{min}\{B^s_\U(F)|\DE(\U)=\nabla(F)\}, & \textrm{when $K^1(F)\neq\emptyset$},\\
    \DE(F)^s&=\,\,B^s_\de(F), \emph{ for all } \de\geq\nabla(F), & \textrm{when $K^1(F)=\emptyset$}.
  \end{array}
\right. \]
\end{lem}

\begin{defn}\label{d:T}
Let $F$ be a finite metric space with no focal points, and suppose $\de\geq\nabla(F)$. Define:
\[T_\de(F):=\textrm{min }\{|\U|\,|\,\U\in K^1_\de(F)\}.
\]
Note that $T_\de(F)\leq |F|-1$, by Remark~\ref{r:n-1elements}. Also, $T_\de(F)\geq 2$.
\end{defn}

\begin{cor}
If $F$ has no focal points, then $B^s(F)=   T_{\nabla(F)}(F) \,\, \nabla(F)^s$.
\end{cor}

\begin{lem}\label{l:Bdecrease}
Suppose that $\DE(F)=1$, and $K^1(F)\neq \emptyset$. Let $f(s)$ denote any of the following functions:
\begin{enumerate}
\item{$B^s_\U(F)$, for all $\U\in K^1(F)$,}
\item{$B^s_\de(F)$, for all $\de\geq \nabla(F)$,}
\item{$B^s(F)$.}
\end{enumerate}
Then $f$ is a positive, strictly decreasing function, $f(0)\geq 2$, and $\lim_{s\rightarrow \infty} f(s)=0$.
\end{lem}

\begin{prop}\label{p:B-simi}
Let $\eta:X\rightarrow X'$ be an $(r,\beta)$-H\"{o}lder equivalence, and $F\subseteq X$ a finite space. Then, for all $s\in[0,\infty)$:
\begin{enumerate}
\item{$B^s_{\eta_*(\U)}(\eta(F)) = r^s B^{s\beta} \U(F)$, for all $\U\in K(F)$.}
\item{$B^s_{r\de^\beta}(\eta(F)) = r^s B^{s\beta} \de(F)$, for all $\de\geq \nabla(F)$.}
\item{$B^s(\eta(F)) = r^s B^{s\beta}(F)$.}
\end{enumerate}
\end{prop}

\subsection{Definition of $\mbox{dim} \m(F)$.} We define \emph{finite box-dimension}, $\textrm{dim} \m(F)$, by solving the equation:

\begin{equation}
B^s(F)=\DE(F)^s,
\end{equation}
which has exactly one solution $s_0^b\in(0,\infty)$, precisely when $F$ has no focal points. More generally, we have:

\begin{prop}\label{p:BuniqueSol}
Consider the following equations:
\begin{enumerate}
\item{$\DE(F)^s=B^s_\U(F)$, for all $\U\in K(F)$,}
\item{$\DE(F)^s=B^s_\de(F)$, for all $\de\geq \nabla(F)$,}
\item{$\DE(F)^s=B^s(F)$.}
\end{enumerate}
Then, in each of these cases, the equation has a unique solution iff $F$ has no focal points. When this is the case, the solutions are positive real numbers, and will be denoted, respectively, $s_\U^b,\,s_\de^b$, and $s_0^b$. Moreover,
\begin{equation}
s^b_0=\frac{\ln T_{\nabla(F)}(F)}{\ln\frac{\DE(F)}{\nabla(F)}}, \,\,\,\,\,\,\,
\textrm{    and  }\,\,\,\,\,\, s^b_\U=\frac{\ln|\U|}{\ln\frac{\DE(F)}{\DE(\U)}}.
\end{equation}
\end{prop}

\begin{lem}
Suppose $F$ has no focal points. Then
\begin{enumerate}
\item{$s_0^b=s_{\nabla(F)}^b$.}
\item{$s_0^b=\emph{max}\{s_\de^b|\de\geq \nabla(F)\}$.}
\item{$s_\de^b=\emph{min}\{s_\U^b|\U\in K^1_\de(F)\}$, for all $\de\geq \nabla(F)$.}
\end{enumerate}
\end{lem}

\begin{defn}\label{d:DIMfB}
For a finite, non-empty subset $F\subseteq (X,d)$, we define
\[ \textrm{dim} \m(F) := \left\lbrace
  \begin{array}{c l}
    0 & \textrm{if $|F|=1$},\\
    \infty & \textrm{if } K^1(F)=\emptyset,\\
    s_0^b & \textrm{if } K^1(F)\neq\emptyset
    .
  \end{array}
\right. \]
\end{defn}
We can summarize our results so far as follows:
\begin{thm}\label{t:DIMfB}
Let $F$ be a non-empty, finite set. Then $\mbox{dim} \m(F)$ is a positive real number if and only if $F$ has no focal points; it is infinity if and only if $F$ has focal points, and it is zero when $F$ has one element. When $F$ has no focal points,
\begin{equation}
\emph{dim} \m(F)=\frac{\ln T_{\nabla(F)}(F)}{\ln \frac{\DE(F)}{\nabla(F)}}\cdot
\end{equation}
\end{thm}

\begin{thm}
Let $\eta:X\rightarrow X'$ be an $(r,\beta)$-H\"{o}lder equivalence, and $F\subseteq X$ a finite space. Then
\begin{equation}\label{e:dimfBSim}
\beta \cdot\emph{dim} \m(\eta(F))=\emph{dim} \m(F).
\end{equation}
\end{thm}

\section{Bounds.}\label{s:5}

In this section we collect technical results that are useful when computing finite dimension. Most of the results are classical ones adapted to the present situation. We start with the relationship between finite Hausdorff and finite box dimension.

\begin{lem}\label{l:H-Brelation}
Suppose $F$ has no focal points, $\de\geq \nabla(F)$, and $s\in [0,\infty)$. Then:
\begin{enumerate}
\item{$T_\de(F) \nabla(F)^s\leq B^s_\U(F)\leq |\U|\,\de^s,\,\,\forall \, \U\in K^1_\de(F).$}
\item{$T_\de (F)\nabla(F)^s\leq B^s_\de(F)\leq T_\de(F)\, \de^s$.}
\item{$H^s_\U(F)\leq B^s_\U(F), \,\,\forall\, \U\in K^1(F)$.}
\item{$H^s_\de(F)\leq B^s_\de(F)$.}
\item{$H^s(F)\leq B^s(F).$}
\end{enumerate}
\end{lem}

\begin{proof}
(i) This follows from the definitions and the inequalities $\nabla(F)\leq \DE(\U)\leq \de$, valid for all $\U\in K^1_\de(F)$. (ii) Both inequalities follow from (i) and the definition of $B^s_\de$. (iii) Let $\U=\{U_1,\dots,U_n\}\in K^1(F)$. Then
\[H^s_\U(F)=\sum_{i=1}^n \DE(U_i)^s\leq \sum_{i=1}^n \DE(\U)^s= |\U|\,\DE(\U)^s=B^s_\U(F),\,\textrm{ as desired.}\]
(iv) Using (iii) and the fact that $K^1(F)\neq \emptyset$, we have
\[H^s_\de(F):=\textrm{min }\{H^s_\U(F)|\,\U\in K^1_\de(F)\} \leq \textrm{min }\{B^s_\U(F)|\,\U\in K^1_\de(F)\}:=B^s_\de(F),\]
as was to be proved. The proof of (v) is similar, but one uses (iv) instead of (iii). This completes the proof of the lemma.
\end{proof}

\begin{cor}
For any $\U\in K^1_{\nabla(F)}(F)$, $H^s(F)\leq |\U|\,\,\nabla(F)^s$.
\end{cor}

\begin{proof}
Given that $F$ has no focal points, for any such $\U$, $H^s(F)\leq H^s_\U(F)\leq B^s_\U(F)$, by (iii) of the lemma.
\end{proof}

\begin{prop}
Let $F$ be a finite metric space. Then,
\begin{equation}\label{e:dimf-fB}
\emph{dim} \hh(F)\leq \emph{dim} \m(F) \leq \frac{\ln (|F|-1)}{\ln \frac{\DE(F)}{\nabla(F)}}\cdot
\end{equation}
\end{prop}

\begin{proof}
Clearly, the first inequality holds (with equality) when $F$ has only one element, or when it has focal points. So suppose $F$ has no focal points. Since both $H^s,\,B^s$ are invariant under similarities, we can assume, without loss of generality, that $\DE(F)=1$. In this case, the desired inequality follows from the fact that $H^s(F)\leq B^s(F)$, proved in Lemma~\ref{l:H-Brelation}. The last inequality follows from Theorem~\ref{t:DIMfB} and Definition~\ref{d:T}. This completes the proof.

\end{proof}

\subsection{Locally uniform spaces.}\label{s:locUnif} These are spaces for which the two finite dimensions we introduced coincide.

\begin{defn}
A finite metric space is called \emph{locally uniform} when \[\de(F)=\nabla(F).\]
Equivalently, when $\nu_F$ is constant.
\end{defn}

\begin{prop}\label{p:dimLocUnif}
If $F$ has no focal points and is locally uniform, then:
\[  \emph{dim} \hh(F) = \emph{dim} \m(F)\cdot \]
Consequently,
\[T_{\nabla(F)}(F)=\big(\DE(F)/\nabla(F)\big)^{\emph{dim} \hh(F)}\cdot\]
\end{prop}

\begin{proof}
Recall the notation $a_1,\dots,a_k$ introduced just before equation (\ref{e:H}). In general, $\de(F)\leq a_1<a_k\leq \nabla(F)$. Our hypothesis imply $k=1$, and $\de(F)= a_1= \nabla(F)$. The proposition follows from Corollary~\ref{c:boundsU}.
\end{proof}

\begin{exa} \label{ex:double}
Let $F$ be an arbitrary finite metric space. Consider its double, $D_x(F)$, defined as follows. Abstractly, it is the product of $F$ with $\{0,1\}$, where $d(0,1)=x$, with the product metric. More concretely, we can assume $F\subseteq \R^n$. Then
\[D_x(F):=\{(b_i,\varepsilon)\in\R^{n+1}|b_i\in F,\,\,\varepsilon=0,x\} \cdot\]
It is easy to see that $D_x(F)$ is locally uniform when $x<\de(F)$. In this case, $T_{\nabla(D_x(F))}(D_x(F))=|F|$, $\nabla(D_x(F))=x$, and $\DE(D_x(F))=\sqrt{\DE(F)^2+x^2}$. By Proposition~\ref{p:dimLocUnif},
\[\textrm{dim} \hh(D_x(F))= \frac{\ln |F|}{\ln\sqrt{1+\big(\frac{\DE(F)}{x}\big)^2}} \cdot\]
\end{exa}

\subsection{Mass distributions.} Mass distributions are used in the classical theory to obtain lower bounds to the Hausdorff dimension. A \emph{mass distribution} is a function $\mu:F\rightarrow [0,\infty)$. We extend to subsets $F'\subseteq F$ by
\[\mu(F'):=\sum_{x\in F'}\mu(x).\]

\begin{lem}
For any family $\{U_i\} {i=1}^m$ of subsets of $F$, we have:
\[\mu(\bigcup_{i=1}^m U_i)\leq \sum_{i=1}^m \mu(U_i).\]
\end{lem}

\begin{proof}
Obvious.
\end{proof}

\begin{prop}\label{t:MASSdist}
Let $\mu$ be a mass distribution on a finite set $F$ with no focal points. Suppose there exist $c>0,\,s>0$, such that $\mu(U)\leq c\,\Delta(U)^s$, for all $U\subseteq F$ with $\Delta(U)\leq \nabla(F)$, and $|U|\geq 2$. Then $\mu(F)\leq c H^s(F)$. If, moreover, $c\,\Delta(F)^s\leq \mu(F)$, then \[s\leq \emph{dim} \hh(F).\]
\end{prop}
\begin{proof}
Let $\U=\{U_1,\dots,U_n\}\in K^1_{\nabla(F)}(F)$, be arbitrary. By hypothesis, $\mu(U_i)\leq c\,\Delta(U_i)^s$. Hence,
\[\mu(F)=\mu(\bigcup U_i)\leq\sum_{i=1}^n\mu(U_i)\leq c\sum_{i=1}^n\Delta(U_i)^s=cH^s_\U(F).\]
This readily implies that $\mu(F)\leq cH^s(F)$, as was to be proved. If we also know that $c\,\Delta(F)^s\leq \mu(F)$, then $\Delta(F)^s\leq H^s(F)$, which shows that $s\leq \textrm{dim} \hh(F)$. This completes the proof.
\end{proof}

\section{Computations.} \label{s:6}
We collect first results of a more or less general nature, and then compute several examples. We begin by showing that every positive real number is the dimension of some finite metric space.

\begin{thm}\label{t:DIMrealiz}
For every $t\in[0,\infty]$ there exist finite spaces $F_t$, such that

\[\emph{dim} \hh(F_t)=t=\emph{dim} \m(F_t).\]
\end{thm}

\begin{proof} We construct a family $F_t$ of locally uniform spaces, so that both dimensions coincide. For $t=0$ [resp. $t=\infty$] we can take $F_t$ to be any singleton [resp. any two-point set]. Suppose now that $t$ is a positive real number, and consider first the case where $t\in[1,\infty)$. For $\eps\in (0,1/2]$, define $A(\eps\emph{}):=\{a_1,a_2,a_3\}\subseteq\mathbb{R}^2$, where $a_1=(0,0)$, $a_2=(1,0)$, and $a_3=\Big(\frac  {1}{2}, \frac  {1}{2}\sqrt{3-8\eps+4\eps^2}\Big)$. We have $d(a_1,a_2)=1$, and $d(a_1,a_3)=d(a_2,a_3)=1-\eps$. $A(\eps)$ is locally uniform, since $\de(A(\eps))=1-\eps=\nabla(A(\eps))$. Clearly, $T_{\nabla(A(\eps))}(A(\varepsilon))=2$, and $\DE(A(\eps))=1$. By Proposition~\ref{p:dimLocUnif},
\[
\textrm{dim} \hh(A(\eps))= \textrm{dim} \m(A(\eps))= \frac{\ln 2}{\ln[\frac  {1}{1-\eps}]}\,\,\cdot
\]
Setting $\eps_t:=1-2^{-1/t}$, and $F_t:=A(\eps_t)$, we have $\textrm{dim} \hh(F_t)=t$, as desired.

Suppose now that $t\in(0,1)$. Set $x(t):=(4^{1/t}-1)^{-1/2}$, and use the double $D_{x(t)}(F)$ of Ex.~\ref{ex:double}, for $F:=\{0,1\}\subseteq \R$. Since $x(t)\in (0,1/\sqrt{3})$, the double is locally uniform, and
\[
\textrm{dim} \hh(D_{x(t)}(F))= \textrm{dim} \m(D_{x(t)}(F))) =t,
\]
as desired. The proof is complete.
\end{proof}

\begin{exa}\label{ex:dimLin}
Let $L_n\subseteq \mathbb{R}^1$ denote a set with $n$ equally  spaced points. Then $L_n$ is locally uniform and, for $n\geq3$, we have:
\[\textrm{dim} \hh(L_{2k})=\frac{\ln k}{\ln(2k-1)};\hspace{5mm} \textrm{dim} \hh(L_{2k+1})=\frac{\ln (k+1)}{\ln (2k)}\cdot\]
Indeed, if the distance between consecutive points is $x>0$, then $\de(L_n)=x=\nabla(L_n)$, $\DE(L_n)=(n-1)x$, and $T_{\nabla(L_n)}(L_n)$ equals $k$ [resp. $k+1$], for $n=2k$ [resp. $n=2k+1$]. Applying Proposition~\ref{p:dimLocUnif}, the result follows. Note that $\lim_{k\to\infty}\mbox{dim} \hh(F_k)= 1$ (the sequence is strictly increasing for $k\geq 4$).
\end{exa}

\begin{lem}\label{l:dim<}
Let $F$ denote a finite metric space with at least 3 elements, and let $t$ be a positive real number. Then the following conditions are equivalent:
\begin{enumerate}
\item{$\emph{dim} \hh(F)<t$,}
\item{$H^t(F)<\DE(F)^t$,}
\item{$\exists\,\U=\{U_1,\dots,U_m\}\in K^1_{\nabla(F)}(F)$, such that
    \[\sum_{i=1}^m \DE(U_i)^t<\DE(F)^t.\]}
\end{enumerate}
\end{lem}

\begin{proof}
We may assume, without loss of generality, that $\DE(F)=1$. Set $s_0= \mbox{dim} \hh(F)$, so that $H^{s_0}(F)=1$. By Lemma~\ref{l:decrease}, if $s_0<t$, then $H^t(F)<1$, as desired. We show now that (ii) implies (iii). Given $t$, we can find $\U\in K^1_{\nabla(F)}(F)$ such that $H^t(F)= H^t_\U(F)=\sum \DE(U_i)^t$, as required. Suppose now that (iii) holds. Given $\U\in K^1_{\nabla(F)}(F)$ satisfying $H^t_\U(F)<1$, we have $H^t(F)\leq H^t_\U(F)<1$. Hence, $s_0<t$. This completes the proof.
\end{proof}

\begin{thm}\label{t:DIMsubsetR}
Let $F$ be a subset of $\mathbb{R}^1$. Then
\begin{enumerate}
\item{If $|F|=3$, then $\emph{dim} \hh(F)=1$, and $\emph{dim} \m(F)\geq 1$. }
\item{If $|F|\geq 4$, then $\emph{dim} \hh(F)<1$.}
\end{enumerate}
\end{thm}

\begin{proof}
Assume $F=\{a_0,a_1,\dots,a_n\}$, and, without loss of generality, that $a_0=0$, and $0<a_1<\cdots<a_n$. Let $y_i:=a_i-a_{i-1}>0$, for $i=1,\dots,n$. We have $\DE(F)=a_n=y_1+\cdots+y_n$.

We prove (i). When $|F|=3$, $\nabla(F)=\textrm{max}\{y_1,y_2\}=y_2$, say, and $K^1(F)=K^1_{\nabla(F)}(F)=\{\U\}$, with $\U=\{\{a_0,a_1\},\{a_1,a_2\}\}$. Hence $H^1(F)=y_1+y_2$, and the equation $H^1(F)=\DE(F)$ yields $\textrm{dim} \hh(F)=1$. On the other hand, $B^s(F)=2y_2^s$, and $ \textrm{dim} \m(F)=\ln 2/\ln(1+(y_1/y_2))\leq 1.$ The proof of (i) is complete.

Consider (ii). Note that $\nu_F(a_1)= y_1,\,\nu_F(a_n)=y_n$, and $\nu_F(a_i)=\textrm{min }\{y_i,y_{y+1}\}$. Suppose $\nabla(F)=y_k$, and consider the 2-covering ($n\geq 3$):
\[\U:=\{\{a_0,a_1\},\{a_1,a_2\},\dots, \{a_{n-1},a_n\}\}.\]
We distinguish two cases: (a) $y_k=\textrm{max }\{y_i|1\leq i\leq n\}$, and (b) $y_k<y_j$, for some $j\leq n$. In case (a), $\DE(\U)=\textrm{max }\{y_i\}=y_k=\nabla(F)$. Then $\U':=\U\setminus \{a_1,a_2\}$ still covers (because $n\geq 3$), and $\U'\in K^1_{\nabla(F)}(F)$. Clearly, $H^1_{\U'}(F)<\DE(F)$, and the result follows from Lemma~\ref{l:dim<}.

In case (b), we set:
\[\U'':=\U\setminus \{U_\ell\in \U| \DE(U_\ell)> y_k\},\]
where $U_\ell$ stands for $\{a_\ell,a_{\ell+1}\}$. Now, $\U''$ would fail to cover $F$, only if: (1) two consecutive $U_i$ are removed, or (2) $U_1$ or $U_n$ are removed. In case (1), suppose we have removed $U_\ell,U_{\ell+1}$, for some $1<\ell<n$. Then $\nu_F(a_\ell)=\textrm{min }\{y_\ell,y_{\ell+1}\}>\nabla(F)$, a contradiction. To deal with (2), notice that always $y_1,\,y_n\leq \nabla(F)$, so these sets will not be removed. In conclusion, $\U''$ is in $K^1_{\nabla(F)}(F)$, and $H^1_{\U''}(F)<\DE(F)$, and Lemma ~\ref{l:dim<} gives the result. This proves (ii) and completes the proof.
\end{proof}

\begin{rem}
When $|F|>3$, $ \textrm{dim} \m(F)$ can be larger or smaller than 1.
\end{rem}

\begin{cor}
Let $F$ be a three-point subset of $\R^n$. Then $\mbox{dim} \hh(F)=1$ if and only if $F$ is collinear.
\end{cor}

\begin{proof}
Suppose $F$ has dimension one. Let $a\leq b\leq c$ denote the three pairwise distances between $F$'s elements. Since $c=\DE(F)$, and $b=\nabla(F)$, we see that $b<c$. It follows that $H^s(F)=a^s+b^s$. By hypothesis, $a+b=H^1(F)=\DE(F)^1=c$, hence the points are collinear. The converse follows from Theorem~\ref{t:DIMsubsetR}. The proof is complete.
\end{proof}

\begin{exa}\label{e:metric}
Let $F\subseteq \R^2$ consist of the points $(0,0),(0,3),(4,0)$. We let $F_2$ [resp. $F_1,\,F_\infty$] denote $F$ with the Euclidean [resp. $\ell_1$, $\ell_\infty$] metric. Then $\mbox{dim} \hh(F_2)=2< \ln 2/\ln(5/4)=\mbox{dim} \m(F_2)$; and $\mbox{dim} \hh(F_1)=1< \ln 2/\ln (7/4)= \mbox{dim} \m(F_1)$. But $\mbox{dim} \hh(F_\infty)=\infty=\mbox{dim} \m(F_\infty)$ because $F_\infty$ has a focal point.
\end{exa}

\begin{exa}\label{ex:fold2}
(Continues Ex.~\ref{ex:fold}). By Theorem~\ref{t:dimHolder}, $\mbox{dim} \hh(F_n')=2\cdot\mbox{dim} \hh(F_n)$. Using Ex.~\ref{ex:dimLin},
\[\mbox{dim} \hh(F_{2k}')=\frac{\ln k^2}{\ln(2k-1)};\qquad
\mbox{dim} \hh(F_{2k+1}')=\frac{\ln (k+1)^2}{\ln(2k)}\cdot\]
Note that the sequence $\mbox{dim} \hh(F_{3}')=2$, $\mbox{dim} \hh(F_{4}')\approx 1.26$, $\mbox{dim} \hh(F_{5}')\approx 1.58,\dots$, converges: $\mbox{dim} \hh(F_{n}')\nearrow 2$, as ${n\to \infty}$.
\end{exa}

While the classical Hausdorff dimension is well-behaved with respect to H\"{o}lder transformations, $\mbox{dim} \hh(-)$ is not. For instance, for a function $\eta:X\rightarrow X'$, the following assertions hold (see Falconer~\cite{Falconer}):
\begin{enumerate}
\item{\emph{If $\eta$ satisfies an $(r,\beta)$-H\"{o}lder condition, then}
\begin{equation}\nonumber
\beta \cdot\mbox{dim}_\hau (\eta(X))\leq \mbox{dim}_\hau (X)
\end{equation}}
\item{\emph{If $\eta$ is bi-Lipschitz, then}
\begin{equation}\nonumber
\mbox{dim}_\hau (\eta(X))= \mbox{dim}_\hau (X)
\end{equation}}
\end{enumerate}

Lemma~\ref{l:simNOcond} suggests that these results cannot hold for $\mbox{dim} \hh(-)$. The example below shows this for (i). We leave it to the reader to find examples where (ii) fails.

\begin{exa}\label{e:noHolder}
Let $F:=\{x_1,\dots,x_4\}\subseteq \R$, where $x_1=0,\,x_2=1,\,x_3=b\geq 1,\,x_4=b+1$. Let $F':=\{y_1,y_2,y_3\}\subseteq \R$, where $y_i=x_i\,(i=1,2,3)$. Define $\eta:F\rightarrow F'$ by $\eta(x_i):= y_i$, for $i=1,2,3$, and $\eta(x_4):=y_3$. Clearly $\eta$ is a 1-similarity, and $\eta(F)=F'$. However, $\mbox{dim} \hh(\eta(F))=\mbox{dim} \hh(F')=1$, while $\mbox{dim} \hh(F)<1$, both claims by Theorem~\ref{t:DIMsubsetR}. Thus, $\mbox{dim} \hh(\eta(F))>\mbox{dim} \hh(F)$, contrary to (i) above.
\end{exa}

\begin{exa}\label{ex:Cantor} \textsc{[Cantor set]}
This example is related to the classical Cantor set $C$. Define a sequence of finite spaces $L_n\subseteq \R$, starting with $L_0:=\{0\}$. Next, add a point to $L_0$, at distance 2/3, to obtain $L_1$. For $L_n$, start from $L_{n-1}$, and add a point at distance $2/3^n$ to the right of every point of $L_{n-1}$. One can see that $|L_n|=2^n$, $\de(L_n)=2/3^n=\nabla(L_n)$, so that the $L_n$ are locally uniform. Moreover, $\DE(L_n)=(3^n-1)/3^n$, and $T_{\nabla(L_n)}(L_n)=2^{n-1}$. Using Proposition~\ref{p:dimLocUnif},
\[
\textrm{dim} \hh(L_n)= \textrm{dim} \m(L_n) =\frac{\ln 2^{n-1}} {\ln \big(\frac{3^n-1}{2} \big)}\cdot
\]

The attentive reader will have noticed the following convergence properties: $L_n\rightarrow C$ in the Hausdorff metric (see the next section for more details), and
\[
\lim_{n\rightarrow \infty}\textrm{dim} \hh(L_n)= \frac{\ln 2} {\ln 3} =\textrm{dim}_\hau (C)\cdot
\]
\end{exa}

\begin{exa} \textsc{[Square of Cantor set]}
Consider the cartesian square of the previous example, $L_n^2\subseteq I\times I\subseteq \R^2$. Thus, $|L_n^2|=2^{2n}$, $\DE(L_n^2)=\sqrt{2}\DE(L_n)=\sqrt{2}(3^n-1)/3^n$, $\de(L_n^2)=\de(L_n)=2/3^n=\nabla(L_n^2)$. Finally, $T_{\nabla(L_n^2)}(L_n^2)= |L_n|T_{\nabla(L_n)}(L_n)= 2^{2n-1}$. Again, using Proposition~\ref{p:dimLocUnif},
\[
\textrm{dim} \hh(L_n^2)= \textrm{dim} \m(L_n^2) =\frac{\ln 2^{2n-1}} {\ln (\frac{\sqrt{2}(3^n-1)}{2})}\cdot
\]
As in the previous example, $L_n^2\rightarrow C^2$, and
\[
\lim_{n\rightarrow \infty}\textrm{dim} \hh(L_n^2)= 2\frac{\ln 2} {\ln 3} =\textrm{dim}_\hau (C^2) \cdot
\]
\end{exa}

\begin{exa} \textsc{[Sierpinski triangle]}. Construct a sequence of finite spaces $L_n$, as follows. $L_0$ consists of a single point, say the origin of $\R^2$. Choose two directions, one in the direction of the $x$-axis, the other forms a 60 degree angle with the $x$-axis, and points towards the first quadrant. To build $L_1$, start with $L_0$ and add two points in the given directions, at distance $1/2$. Inductively, construct $L_n$ from $L_{n-1}$, by adding two points in the given directions to each point of $L_{n-1}$, at distance $1/2^n$. The following properties are easy to check: $|L_n|=3^n,\,\,\DE(L_n)=(2^n-1)/2^n$, and $\de(L_n)=1/2^n=\nabla(L_n)$. Finally, $T_{\nabla(L_n)}(L_n)=3^{n-1}$. Hence,
\[
\textrm{dim} \hh(L_n)= \textrm{dim} \m(L_n) =\frac{\ln 3^{n-1}} {\ln (2^n-1)}\cdot
\]
The reader can check that $L_n\rightarrow ST$, where $ST$ stands for the Sierpinski triangle, and
\[
\lim_{n\rightarrow \infty}\textrm{dim} \hh(L_n)= \frac{\ln 3} {\ln 2} =\textrm{dim}_\hau (ST) \cdot
\]

Proceeding in a similar way, but starting with three directions in $\R^3$, each forming a 60 degree angle with the other, one can construct a sequence of finite spaces $L_n$, related to the Sierpinski tetrahedron $STh$. It is not difficult to check that: $|L_n|=4^n,\,\,\DE(L_n)=(2^n-1)/2^n$, and $\de(L_n)=1/2^n=\nabla(L_n)$. Finally, $T_{\nabla(L_n)}(L_n)=4^{n-1}$. Hence,
\[
\textrm{dim} \hh(L_n)= \textrm{dim} \m(L_n) =\frac{\ln 4^{n-1}} {\ln (2^n-1)}\cdot
\]
As before, $L_n\rightarrow STh$, and
\[
\lim_{n\rightarrow \infty}\textrm{dim} \hh(L_n)= \frac{\ln 4} {\ln 2}=2 =\textrm{dim}_\hau (STh) \cdot
\]
\end{exa}

\begin{exa} \textsc{[Cantor carpet]}. Recall the classical Cantor carpet $CC$, and let $Q_0$ denote the unit square $I^2\subseteq\R^2$. Divide $Q_0$ into nine subsquares of side $1/3$, and remove the interior of the central one; call this set $Q_1$. Let $Q_{1,i}$ ($i=1,\dots,8$) denote the eight remaining subsquares of $Q_  {1}$. To obtain $Q_2$, replace each $Q_{1,i}$, by $Q_1$ scaled by $1/3$. Thus $Q_2$ has nine holes: a central square of side $1/3$, and eight squares of side $1/3^2$, surrounding the large one.

In general, $Q_{n+1}$ is obtained from $Q_n$ by replacing each $Q_{n}\cap Q_{1,i}$ ($i=1,\dots,8$), by $Q_{n}$ scaled by $1/3$. Thus, $Q_{n+1}$ has $8^{n}$ holes that are squares of side $1/3^{n}$ (the holes of smallest size in $Q_{n+1}$).

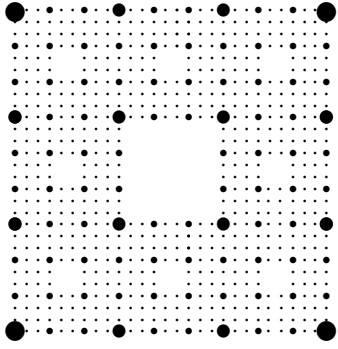
\begin{figure}[h]
\begin{center}%
\resizebox{125pt}{125pt}{
\begin{picture}(250,250)(-100,100)
\multiput(-100,100)(0,9){28}{\circle*   {1.5}}
\multiput(-91,100)(0,9){28}{\circle*   {1.5}}
\multiput(-82,100)(0,9){28}{\circle*   {1.5}}
\multiput(-73,100)(0,9){28}{\circle*   {1.5}}
\multiput(-46,100)(0,9){28}{\circle*   {1.5}}
\multiput(-37,100)(0,9){28}{\circle*   {1.5}}
\multiput(-28,100)(0,9){28}{\circle*   {1.5}}
\multiput(-19,100)(0,9){28}{\circle*   {1.5}}
\multiput(-64,100)(0,9){4}{\circle*   {1.5}}
\multiput(-55,100)(0,9){4}{\circle*   {1.5}}
\multiput(-64,154)(0,9){7}{\circle*   {1.5}}
\multiput(-55,154)(0,9){7}{\circle*   {1.5}}
\multiput(-64,235)(0,9){7}{\circle*   {1.5}}
\multiput(-55,235)(0,9){7}{\circle*   {1.5}}
\multiput(-64,316)(0,9){4}{\circle*   {1.5}}
\multiput(-55,316)(0,9){4}{\circle*   {1.5}}
\multiput(-10,100)(0,9){10}{\circle*   {1.5}}
\multiput(-1,100)(0,9){10}{\circle*   {1.5}}
\multiput(8,100)(0,9){10}{\circle*   {1.5}}
\multiput(-10,262)(0,9){10}{\circle*   {1.5}}
\multiput(-1,262)(0,9){10}{\circle*   {1.5}}
\multiput(8,262)(0,9){10}{\circle*   {1.5}}
\multiput(17,100)(0,9){4}{\circle*   {1.5}}
\multiput(26,100)(0,9){4}{\circle*   {1.5}}
\multiput(35,100)(0,9){10}{\circle*   {1.5}}
\multiput(17,154)(0,9){4}{\circle*   {1.5}}
\multiput(26,154)(0,9){4}{\circle*   {1.5}}
\multiput(35,154)(0,9){4}{\circle*   {1.5}}
\multiput(17,262)(0,9){4}{\circle*   {1.5}}
\multiput(26,262)(0,9){4}{\circle*   {1.5}}
\multiput(35,154)(0,9){4}{\circle*   {1.5}}
\multiput(17,316)(0,9){4}{\circle*   {1.5}}
\multiput(26,316)(0,9){4}{\circle*   {1.5}}
\multiput(35,316)(0,9){4}{\circle*   {1.5}}
\multiput(35,262)(0,9){6}{\circle*   {1.5}}
\multiput(26,316)(0,9){4}{\circle*   {1.5}}
\multiput(35,316)(0,9){4}{\circle*   {1.5}}
\multiput(44,100)(0,9){10}{\circle*   {1.5}}
\multiput(53,100)(0,9){10}{\circle*   {1.5}}
\multiput(44,262)(0,9){10}{\circle*   {1.5}}
\multiput(53,262)(0,9){10}{\circle*   {1.5}}
\multiput(62,100)(0,9){28}{\circle*   {1.5}}
\multiput(71,100)(0,9){28}{\circle*   {1.5}}
\multiput(80,100)(0,9){28}{\circle*   {1.5}}
\multiput(89,100)(0,9){28}{\circle*   {1.5}}
\multiput(97,100)(0,9){4}{\circle*   {1.5}}
\multiput(107,100)(0,9){4}{\circle*   {1.5}}
\multiput(97,154)(0,9){7}{\circle*   {1.5}}
\multiput(106,154)(0,9){7}{\circle*   {1.5}}
\multiput(98,235)(0,9){7}{\circle*   {1.5}}
\multiput(107,235)(0,9){7}{\circle*   {1.5}}
\multiput(98,316)(0,9){4}{\circle*   {1.5}}
\multiput(107,316)(0,9){4}{\circle*   {1.5}}
\multiput(116,100)(0,9){28}{\circle*   {1.5}}
\multiput(125,100)(0,9){28}{\circle*   {1.5}}
\multiput(133,100)(0,9){28}{\circle*   {1.5}}
\multiput(142,100)(0,9){28}{\circle*   {1.5}}
\multiput(-100,100)(0,242){2}{\circle*{15}}
\multiput(142,100)(0,242){2}{\circle*{15}}
\multiput(-100,181)(0,81){2}{\circle*{10}}
\multiput(142,181)(0,81){2}{\circle*{10}}
\multiput(-19,100)(0,81){4}{\circle*{10}}
\multiput(62,100)(0,81){4}{\circle*{10}}
\multiput(-100,181)(0,81){2}{\circle*{10}}
\multiput(142,181)(0,81){2}{\circle*{10}}
\multiput(-73,100)(0,27){10}{\circle*{5}}
\multiput(-46,100)(0,27){10}{\circle*{5}}
\multiput(89,100)(0,27){10}{\circle*{5}}
\multiput(116,100)(0,27){10}{\circle*{5}}
\multiput(-100,127)(0,27){8}{\circle*{5}}
\multiput(-19,127)(0,27){8}{\circle*{5}}
\multiput(62,127)(0,27){8}{\circle*{5}}
\multiput(142,127)(0,27){8}{\circle*{5}}
\multiput(8,100)(0,27){4}{\circle*{5}}
\multiput(35,100)(0,27){4}{\circle*{5}}
\multiput(8,262)(0,27){4}{\circle*{5}}
\multiput(35,262)(0,27){4}{\circle*{5}}
{
\linethickness{.7pt}
}
{\linethickness{.25pt}
}
{\linethickness{.7pt}
}
{
\linethickness{.25pt}
}
{
\linethickness{.25pt}
}
{
\linethickness{.25pt}
}
{
\linethickness{.25pt}
}
{
\linethickness{.25pt}
}
\end{picture}
}
\end{center}
\caption{Cantor Carpet: $L_0\subseteq L_1\subseteq L_2\subseteq L_3$. %
         \label{figure:c-lambda}}
\end{figure}

We approximate $CC$ by an increasing sequence of finite spaces $L_n$  (see Figure 1). Let $L_0$ consist of the four vertices of the unit square. To obtain $L_1$ we replace each $Q_{1,i}$ ($i=1,\dots,8$), by $L_0$ scaled by $1/3$; so $L_1$ consists of 16 points. In general, $L_{n+1}$ is obtained from $L_n$ by replacing each $L_{n}\cap Q_{1,i}$ ($i=1,\dots,8$), by $L_{n}$ scaled by $1/3$.

Let's compute recursively the cardinality of the $L_n$. The four corner portions of $L_{n+1}$ have $4|L_{n}|$ elements, and the rest (four portions which, together with the missing central part, build a "cross") contributes with $4(|L_n|-2(3^n+1))$; i.e. each of these four squares contributes with $|L_n|$ minus two sides (which they have in common with each corner portion). So we have the recursive formula:
\begin{equation}\label{e:cardCC}
|L_0|=4,\quad  |L_{n+1}| = 4|L_n|+4(|L_n|-2(3^n+1)) \quad (n\geq 0)
\end{equation}

We claim that
\begin{equation}\label{e:CC}
8^{n}\leq |L_n|\leq 2\cdot 8^{n}, \quad \mbox{for }\quad n\geq 2.
\end{equation}
To prove (\ref{e:CC}), write $|L_n|=8^n+k_n$ ($n\geq 0$). We show that $k_n>0$ for all $n$, by establishing the following claims. First, $8k_n \geq 22(3^n+1)$, which is valid for $n\geq 2$, and second, $k_n\leq 8^n$, valid for $n\geq 1$. The inequalities (\ref{e:CC}) follow directly from these claims.

For the first claim, let $n\geq 2$, and consider the inductive step. Using (\ref{e:cardCC}), $k_{n+1}=8[k_n-(3^n+1)]$. By the induction hypothesis, $8k_{n+1}\geq 8\cdot 22\cdot(3^n+1)-8^2(3^n+1)=112\cdot(3^n+1)$. The desired inequality $8k_{n+1}\geq 22\cdot(3^{n+1}+1)$ follows, since $22\cdot(3^{n+1}+1)=66\cdot(3^n+1)-44$, and $112\cdot(3^n+1)\geq 66\cdot(3^n+1)-44$.

The remaining inequality, $k_n\leq 8^n$ ($n\geq 1$), follows easily by induction. For the inductive step, recall $k_{n+1}=8[k_n-(3^n+1)]\leq 8\cdot 8^n-8\cdot(3^n+1)\leq 8^{n+1}$. This establishes (\ref{e:CC}).

The $L_n$ have the following properties: $\DE(L_n)=\sqrt2$,  $\de(L_n)=1/3^{n}=\nabla(L_n)$. Moreover, $T_{\nabla(L_n)}(L_n)=|L_n|/2$. Hence,
\[
\textrm{dim} \hh(L_n)= \textrm{dim} \m(L_n) =\frac{\ln (|L_n|/2)} {\ln (3^{n}\sqrt2)}\cdot
\]
Clearly, $L_n\rightarrow CC$ and, in view of (\ref{e:CC}),
\[
\lim_{n\rightarrow \infty}\textrm{dim} \hh(L_n)= \frac{\ln 8} {\ln 3}= \textrm{dim}_\hau (CC) \cdot
\]

\end{exa}

\section{Convergence}\label{s:7}

Let $Z$ be an arbitrary metric space. The set of all closed and bounded subsets of $Z$ with the Hausdorff distance $d _\hau$ is a metric space, denoted $(\FM(Z),d _\hau)$ or, more simply, $\FM(Z)$. In this section we prove two convergence theorems. They give conditions under which, if $F_n$ is a sequence of finite spaces that converges in $\FM(Z)$ to a space $X$, then $\mbox{dim} \hh(F_n)\rightarrow \mbox{dim}_\hau (X)$ [resp. $\mbox{dim} \m(F_n)\rightarrow \mbox{dim} _\boxx(X)$].

\subsection{Preliminaries.}

The next approximation result is well-known, see e.g. \cite{bbi}.

\begin{prop}\label{p:Fapprox}
Every compact subset of a metric space $Z$ is the limit, in $(\FM(Z),d _\hau)$, of a sequence of finite subsets.
\end{prop}

Next, we extend previous definitions to spaces that are not necessarily finite.

\begin{defn}
Let $(X,d)$ be an arbitrary metric space. The diameter of $X$, $\DE(X)$, is $\sup\{d(x,y)|x,y\in X\}$. When $\DE(X)>0$ (equivalently, $X$ has at least two points), we define $\nu_X:X\rightarrow \R$,  by \[\nu_X(x):=\textrm{inf }\{d(x,y)|y\in X,\, x\neq y\},\]
and the constants:
\[\delta(X):=\inf\{\nu_X(x)|x\in X\}\in [0,\infty),\quad \nabla(X):=\sup\{\nu_X(x)|x\in X\}\in [0,\infty].\]
\end{defn}
\begin{rem}
In contrast to the diameter, which is defined for any (non-empty) space, we define $\nu_X,\de(X)$, and $\nabla(X)$, only for spaces with at least two points. As before, notice that $0\leq \de(X)\leq\nabla(X)\le\DE(X)$. However, $\de(X)$ may now be zero. When this is the case, $X$ is infinite (the converse is false, as $X=\N$ shows). The next result clarifies the meaning of $\nabla(X)=0$.
\end{rem}

Recall that we say that $X$ has \emph{no isolated points} if, for all $\eps>0$, for all $x\in X$, there is $y\neq x$ in $X$, such that $d(x,y)<\eps$.

\begin{lem}\label{l:nabla=0}
Let $X$ be an arbitrary metric subspace of $Z$, satisfying $\DE(X)>0$. Then the following conditions are equivalent:
\begin{enumerate}
\item{$\nabla(X)=0$.}
\item{$X\subseteq X'$, where $X'$ denotes the derived set.}
\item{$X$ has no isolated points.}
\item{$X\cap B(x,r)$ is infinite, for all $x\in X$, and all $r>0$.}
\item{Let $x\in X$ be arbitrary, and let $\{x_k\}$ be any sequence in $X$ that converges to $x$. Then there is a sequence $\{x'_k\}$ in $X$, such that
    \begin{enumerate}
    \item{$\{x'_k\}$ converges to $x$,}
    \item{$0<d(x_k,x'_k)<1/k$, for all $k\in \N$,}
    \item{For all $k,\ell\in \N$, if $x'_k=x'_\ell$, then $k=\ell$.}
    \end{enumerate}}
\end{enumerate}
If, in addition, $X$ is closed, then \emph{(ii)} can be replaced by the condition that $X$ is perfect (i.e. $X=X'$). In particular, $X$ is uncountably infinite.
\end{lem}

\begin{proof}
The equivalences (i)--(iv) will be left to the reader. Given (iv) and $x_k\to x$, we define $\{x'_k\}$ inductively. Set $x'_1:=x_1$. Suppose $x'_1,\dots,x'_k$ are constructed, and satisfy (b)--(c). Since $X\cap B(x_{k+1},1/(k+1))$ is infinite, we can find $x'_{k+1}\in X\cap B(x_{k+1},1/(k+1))$ such that $x'_{k+1}\notin \{x'_1,\dots,x'_k,x_{k+1}\}$. Clearly (b)--(c) are satisfied, and the constructed sequence satisfies (a) too, as desired.

We show that (v) implies (iv). Let $x\in X$ be arbitrary. From the constant sequence $x_k:=x$, for all $k$, we obtain, by (v), a sequence satisfying (a)--(c). Then the intersection of $X$ with any ball centered at $x$, will contain a tail of the sequence, which consists of distinct points. Thus, the intersection must be infinite.

When $X$ is closed we have $X'\subseteq X$. This, together with (ii), gives $X=X'$, i.e. $X$ is perfect. The converse is obvious. The proof is complete.
\end{proof}

\begin{lem}\label{l:NUcont}
Let $(X,d)$ be a metric space with $\DE(X)>0$. Then $\nu_X$ is continuous.
\end{lem}

\begin{proof}
Suppose $\eps>0$, and $x_0\in X$ are given. Then there is $x'_0\in X$ such that
$0<d(x_0,x'_0)<\nu_X(x_0)+\eps$. Let $\delta$ satisfy: $0<\delta<\min\{d(x_0,x'_0),\eps\}$. Then, for any $x\in B(x_0,\delta)\cap X$, $d(x,x'_0)\leq d(x,x_0)+d(x_0,x'_0)<\delta+\nu_X(x_0)+\eps$. Note that $x$ is different from $x'_0$, for if $x=x'_0$, then $\delta>d(x,x_0)=d(x'_0,x_0)>\delta$, a contradiction. It follows that
\[\nu_X(x)\leq d(x,x'_0)<\nu_X(x_0)+\eps+\delta<\nu_X(x_0)+2\eps\]
To prove $\nu_X(x_0)<\nu_X(x)+2\eps$, choose $x'\in X$ such that
$0<d(x,x') <\nu_X(x)+\eps$. If $x'\neq x_0$, then $\nu_X(x_0)\leq d(x_0,x')\leq d(x_0,x)+d(x,x')<\de+\nu_X(x)+ \eps\leq \nu_X(x)+2\eps$, as desired. If $x'=x_0$, then $\nu_X(x_0) \leq d(x,x_0)<\nu_X(x)+\eps$. The proof is complete.
\end{proof}

\begin{lem}\label{l:approxDE}
Suppose that $X,Y\subseteq Z$, are arbitrary metric spaces, and $\de>0$. If $d _\hau(X,Y)<\de$, then $|\DE(X)-\DE(Y)|\leq 2\emph{}\de$.
\end{lem}

\begin{proof} For arbitrary $x_1,x_2\in X$, one can find $y_1,y_2\in Y$ such that $d(y_i,x_i)<\de$. It follows that $d(x_1,x_2)<2\de + d(y_1,y_2)\leq 2\de+\DE(Y)$. Hence, $\DE(X)\leq 2\de+\DE(Y)$. The reverse inequality is proved similarly. The proof is complete.
\end{proof}

\begin{lem}\label{l:approxNA}
Suppose that $X,Y\subseteq Z$, are arbitrary compact metric spaces, and $\nabla(X)=0$. Then, for all $\eps >0$, there exists $\de>0$, such that, if $d _\hau(X,Y)<\de$, then $\nabla(Y)< \eps$.
\end{lem}

\begin{proof}
Suppose, for contradiction, that we can find $\eps>0$, a compact space $X$ with $\nabla(X)=0$, and, for all $\de>0$, a compact $Y_\de$, such that
\[d _\hau(X,Y_\de)<\de\quad\mbox{ and } \quad\nabla(Y_\de)\geq \eps.\]
In particular, for all $k\in \N$, there exist compact spaces $Y_k\subseteq Z$, such that
\begin{equation}\label{e:nablaYn}
d _\hau(X,Y_k)<1/k\quad\mbox{ and }\quad\nabla(Y_k)\geq \eps.
\end{equation}
By compactness of $Y_k$ and continuity of $\nu_{_{Y_k}}$ (Lemma~\ref{l:NUcont}), we can find points $y_k\in Y_k$ ($k\in\N$), such that
\begin{equation}\label{e:NU=naYn}
\nu_{_{Y_k}}(y_k)=\nabla(Y_k).
\end{equation}
Since $d _\hau(X,Y_k)<1/k$, $Y_k\subseteq N_{1/k}(X)$, the tubular neighborhood of $X$ of radius $1/k$, defined by $N_{1/k}(X):=\bigcup_{x\in X} B(x,1/k)$. Hence, we can find $x_k\in X$, such that, for all $k\in\N$,
\begin{equation}\label{e:d(xk,yk)<1/k}
d(x_k,y_k)<1/k.
\end{equation}
By compactness of $X$, there exists a subsequence of $\{x_k\}_{k=1}^\infty$, still denoted $\{x_k\}_{k=1}^\infty$, and a point $x_0\in X$, such that
\begin{equation}\label{e:xnCONVx0}
x_k\To x_0.
\end{equation}
Observe that, since $\nabla(X)=0$, by Lemma~\ref{l:nabla=0}(v), we can further assume that $x_k\neq x_j$, whenever $k\neq j$. In particular, $x_k=x_0$, for at most one $k\geq 1$.
From (\ref{e:d(xk,yk)<1/k}) and (\ref{e:xnCONVx0}), we get:
\begin{equation}\label{e:ynCONVx0}
y_k\To x_0.
\end{equation}
By (\ref{e:xnCONVx0}) and the remark following it, we can choose $\ell$ so large that \begin{equation}\label{e:eps'}
0<d(x_0,x_{\ell})<\eps/100.
\end{equation}
We set $\eps':=d(x_0,x_{\ell})$. Choose $M>3$, and define $r:=\eps'/M$. By (\ref{e:eps'}), $r>0$. By (\ref{e:ynCONVx0}), there is $m$ so large that
\begin{equation}\label{e:kBIS}
1/m < r, \quad\mbox{ and }\quad d(x_0,y_{m})<r.
\end{equation}
Thus, by (\ref{e:nablaYn}) $d _\hau(X,Y_{m})<1/m<r$, so that $X\subseteq N_r(Y_{m})$. Hence, there is a point $y'_{m}\in Y_{m}$, such that
\begin{equation}\label{e:d(x,y'm)<r}
d(x_{\ell},y'_{m})<r.
\end{equation}
Now, $d(y'_{m},y_{m})\leq d(y'_{m},x_{\ell})+d(x_{\ell},x_0)+d(x_0,y_{m})<\eps'(1+2/M)<2\eps/100$, where we have used (\ref{e:eps'})--(\ref{e:d(x,y'm)<r}). Hence,
\begin{equation}\label{e:d(y'm,ym)<eps}
d(y'_{m},y_{m})<\eps
\end{equation}
Moreover, we claim that $d(y'_{m},y_{m})>0$. Indeed, using (\ref{e:kBIS}) and (\ref{e:d(x,y'm)<r}), we see that $\eps'=d(x_0,x_{\ell})\leq d(x_0,y_{m})+d(y_{m},y'_{m})+d(y'_{m},x_{\ell})<2r+d(y_{m},y'_{m})$. Hence, $d(y_{m},y'_{m})>\eps'-2r=r(M-2)>r>0$, as desired. It follows that
\begin{equation}\label{e:nuLEQd(,)}
\nu_{_{Y_{m}}}(y_{m})\leq d(y_{m},y'_{m}).
\end{equation}
Then (\ref{e:nablaYn}), (\ref{e:NU=naYn}), (\ref{e:d(y'm,ym)<eps}), and (\ref{e:nuLEQd(,)}), yield:
\begin{equation}\nonumber
\eps\leq \nabla(Y_{m})=\nu_{_{Y_{m}}}(y_{m})\leq d(y_{m},y'_{m})<\eps,
\end{equation}
a contradiction. This concludes the proof.
\end{proof}

Lemma~\ref{l:approxNA} is false without the hypothesis $\nabla(X)=0$, as the following example shows.

\begin{exa} It follows from the lemma that $Y_s\to X$ implies $\nabla(Y_s)\to 0$. We present an example where $Y_s\to X$, but $\nabla(Y_s)$ does not converge to $\nabla(X)$. Recall the double $D_s(F)$ defined in Example~\ref{ex:double}. Clearly, $D_s(F)\to F$, when $s\to 0$. From Example~\ref{ex:double}, we know that for $s$ small enough, $\nabla(D_s(F))=s$, so that $\nabla(D_s(F))\to 0$. Hence, we obtain the desired counterexample by choosing any $F$ with $\nabla(F)>0$. For instance, $F=\{0,1,2,3\}\subset\R$.

\end{exa}

\begin{lem}\label{l:Tlarge}
Suppose $X\subseteq Z$ is compact and $\nabla(X)=0<\Delta(X)$. Let $\{F_k\}$ be a sequence of finite subsets of $Z$ that converges to $X$ in $(\FM(Z),d _\hau)$. Then $\lim_{k\to\infty}T_{\nabla(F_k)}(F_k)=\infty$.
\end{lem}

\begin{proof} Suppose, for contradiction, that
\begin{equation}\label{e:limInf1}
\exists\,\,M\geq 1 \quad\forall\,\, N\quad \exists\,\, k\geq N\,\, : \,\, T_{\nabla(F_k)}(F_k)\leq M.
\end{equation}
Since $\DE(X)>0$, we can choose distinct points $x,x_1\in X$. The condition $\nabla(X)=0$ implies, by Lemma~\ref{l:nabla=0}, that $x$ is not isolated. Applying Lemma~\ref{l:nabla=0}(iv) repeatedly, we can construct, starting from $x_1$, a sequence $\{x_i\}_{i=1}^\infty\subseteq X$, such that:
\begin{equation}\label{e:limInf2}
0<d(x,x_{i+1})<\frac{d(x,x_i)}{2}.
\end{equation}
For $1\leq i<j$, we have $d(x_i,x_j)\geq d(x,x_i)-d(x,x_j)$. Applying (\ref{e:limInf2}) repeatedly:
\begin{equation}\label{e:limInf3}
d(x_i,x_{j})>(2^{j-i}-1)d(x,x_j), \mbox{ for all } i<j\in\N.
\end{equation}

If $d _\hau(F_k,X)<\delta$, there are points $y^k_i\in F_k$ that correspond to the $x_i$, i.e. $d(y^k_i,x_i)<\de$ (both $\de$ and $k$ will be determined presently). It follows that $d(x_i,x_j)< 2\de + d(y^k_i,y^k_j)$ which, together with (\ref{e:limInf3}), yields:
\begin{equation}\label{e:limInf4}
d(y^k_i,y^k_{j})>(2^{j-i}-1)d(x,x_j)-2\de, \mbox{ for all } i<j.
\end{equation}

For $M$ in (\ref{e:limInf1}), set $\alpha_{M+1}:=\min\{d(x,x_i)|i=1,\dots,M+1\}>0$. Choose $\de$ so that:
\begin{equation}\label{e:limInf5}
0<\de<\frac{\alpha_{M+1}}{2^{M+1}}.
\end{equation}
By Lemma ~\ref{l:nabla=0}, there exists $L\geq 1$ such that $0<\nabla(F_k)<(\alpha_{M+1}/2^{M})-2\de$, for all $k\geq L$. For this $L$ we can find, by (\ref{e:limInf1}), $\ell\geq L$ such that $T_{\nabla(F_\ell)}(F_\ell)\leq M$. So, for this $\ell$, we have:
\begin{equation}\label{e:limInf6}
\nabla(F_\ell)<\frac{\alpha_{M+1}}{2^{M}}-2\de,\quad\mbox{ and } \quad T_{\nabla(F_\ell)}(F_\ell)\leq M.
\end{equation}

Define $S:=\{y^\ell_1,\dots,y^\ell_{M+1}\}\subseteq F_\ell$. It follows from (\ref{e:limInf5}) and (\ref{e:limInf4}), that the elements of $S$ are distinct, i.e. $|S|=M+1$. Let $\U\in K_{\nabla(F_\ell)}(F_\ell)$, $\U=\{U_1,\dots,U_n\}$, be a 2-covering of $F_\ell$, with $n=T_{\nabla(F_\ell)}(F_\ell)$. We claim that $|U_i\cap S|\leq 1$, for all $U_i\in\U$. Otherwise, there would be two elements, say $y^\ell_i, y^\ell_j\in U_m\cap S$, for some $i<j\leq M+1$, $1\leq m\leq n$. Then $d(y^\ell_i,y^\ell_j)\leq \DE(U_m)\leq\nabla(F_\ell)< (\alpha_{M+1}/2^M)-2\de$, by (\ref{e:limInf6}). Using (\ref{e:limInf4}) and simplifying, we have
\begin{equation}\nonumber
\alpha_{M+1}/2^M>(2^{j-i}-1)d(x,x_j)\geq d(x,x_j)\geq \alpha_{M+1},
\end{equation}
the last inequality by the definition of $\alpha_{M+1}$. This is a  contradiction, as desired.

Now, since $\U$ covers $F_\ell$, $S=\bigcup_{i=1}^n U_i \cap S$, so that
$M+1=|S|\leq \sum_{i=1}^n |U_i \cap S|\leq n\leq M$, a contradiction. This establishes the lemma.

\end{proof}

\begin{lem}\label{l:lim}
Suppose given three convergent sequences of positive real numbers $\{T_k\}$, $\{\DE_k\}$, and $\{\nabla_k\}$, such that $\nabla_k< \DE_k$ for all $k$, and $T_k\to \infty$, $\DE_k\to \DE>0$, $\nabla_k\to 0$. Then
\begin{equation} \nonumber
\lim_{k\rightarrow \infty} \frac{\ln T_k}{\ln \frac{\Delta_k}{\nabla_k}}=\lim_{k\rightarrow \infty} \frac{\ln T_k}{-\ln \nabla_k}
\end{equation}
\end{lem}
\begin{proof}
The proof is obvious.
\end{proof}

We can summarize these results in the following:
\begin{prop}\label{p:limDIMfB}
Suppose that $X$ is a compact metric space satisfying $\nabla(X)=0<\Delta(X)$, and $\lim_{k\to\infty} F_k=X$, for some sequence of finite metric spaces. Then
\[
\lim_{k\rightarrow \infty} \emph{dim} \m(F_k)=\lim_{k\rightarrow \infty} \frac{\ln T_{\nabla(F_k)}(F_k)}{-\ln \nabla(F_k)}\cdot
\]
\end{prop}

\begin{proof} Let $T_k:=T_{\nabla(F_k)}(F_k)$, $\DE_k:=\DE(F_k)$, and $\nabla_k :=\nabla(F_k)$. By lemmas~\ref{l:approxDE},~\ref{l:approxNA}, and~\ref{l:Tlarge}, the hypothesis of Lemma~\ref{l:lim} are satisfied for all large enough $k$. This gives the result.
\end{proof}

\subsection{Subsets of $\R^n$}

 In this section we exploit special properties of Euclidean space to refine Proposition~\ref{p:Fapprox}. For any $\eps >0$, consider the lattice $\eps\cdot\Z^n\subseteq \R^n$. Using the hyperplanes parallel to the coordinate hyperplanes through each point of the lattice, we obtain a tiling of $\R^n$ by hypercubes. Let $\QQ(\varepsilon)$ denote the set of all these hypercubes, which we simply call cubes. The cubes have side-length $\eps$, and diameter $\eps\sqrt{n}$. Each cube is compact and has $2^n$ corners that lie in $\eps\cdot\Z^n$. Each corner belongs to $2^n$ cubes.

\begin{defn}
For $X\subseteq \R^n$, $\QQ(\eps,X)$ will denote the set of cubes of $\QQ(\eps)$ that have non-empty intersection with $X$.
\end{defn}

\begin{prop}\label{t:FapproxR^n}
For every compact $X\subseteq\R^n\,(n\geq 1)$, and $\varepsilon >0$, there is a finite, locally uniform $F_\varepsilon \subseteq \R^n$, such that $d _\hau(F_\varepsilon, X)\leq \varepsilon\sqrt{n}$. In particular, given $0<c<1$, there is a convergent sequence of finite, locally uniform $F_k \subseteq \R^n$, such that $F_k\to X$ in $(\FM(\R^n),d _\hau)$, and $\nabla(F_{k+1})\geq c\cdot\nabla(F_k)$, for all $k\in\N$.
\end{prop}


\begin{proof}
Define $F_\eps:=\eps\cdot\Z^n\cap \QQ(\eps,X)$. Since $X$ is bounded, $\QQ(\eps,X)$ is finite with, say, $N$ elements. Hence, $F_\eps$ contains no more than $N\cdot 2^n$ points. Observe that if $x\in F_\eps$, then there is at least one $Q\in\QQ(\eps,X)$ such that $x$ is one of the corners of $\QQ$. Hence, $F_\eps$ contains not only $x$, but all $2^n$ corners of $Q$. It follows that $\delta(F_\eps)=\nabla(F_\eps)=\eps$; in other words, $F_\eps$ is locally uniform. Moreover, $d _\hau(X,F_\eps)\leq \eps\sqrt{n}$. For the last part, given $0<c<1$, the sequence $F_k:=F_{c^k}$ satisfies the required conditions. This completes the proof.
\end{proof}

\begin{cor}\label{c:LIMdimH}
Let $X$ and $\{F_k\}$ be as in Proposition~\ref{t:FapproxR^n}, and suppose that
$\nabla(X)=0<\DE(X)$. Then
\[
\lim_{k\rightarrow \infty} \emph{dim} {\hh}(F_k)=\lim_{k\rightarrow \infty} \frac{\ln T_{\nabla(F_k)}(F_k)}{-\ln \nabla(F_k)}\cdot
\]
\end{cor}

\begin{proof}
By Proposition~\ref{p:limDIMfB},
\[
\lim_{k\rightarrow \infty} \emph{dim} {\m}(F_k)=\lim_{k\rightarrow \infty} \frac{\ln T_{\nabla(F_k)}(F_k)}{-\ln \nabla(F_k)}\cdot
\]
By Proposition~\ref{t:FapproxR^n}, the $F_k$ are locally uniform. Hence $\mbox{dim} \hh(F_k)=\mbox{dim} \m(F_k)$, by Proposition~\ref{p:dimLocUnif}. The result follows.
\end{proof}

\begin{lem}\label{l:vert-cubos}
Let $C\subseteq \QQ(\varepsilon)$, and put $F(C):= \varepsilon \cdot\Z^n\cap (\bigcup_{Q\in C}Q)$, the set of all lattice points of the cubes of $C$. Then
\begin{equation}\label{e:lat-cubes}
|F(C)|\geq |C|.
\end{equation}
\end{lem}

\begin{proof}
Consider the linear map $f:\R^n\rightarrow \R$ defined by $f(x):=\sum_{i=1}^n x_i$. Put $M=\mbox{max}\,\{f(x)|x\in F(C)\}$, and let $a\in F(C)$ be a point with $f(a)=M$. We show that there is a unique cube $Q(a)\in C$, such that $a\in Q(a)$.

Let $e_i:=(0,\dots,0,1,0,\dots,0), \, i=1,\dots,n$, denote the unit vectors of $\R^n$, and define points $p_i=a+\varepsilon e_i$, and $p_{n+i}=a-\varepsilon  e_i$, where $i=1,\dots,n$. Observe that these points belong to $\varepsilon \cdot\Z^n$, and $f(p_i)=f(a)\pm \varepsilon f(e_i)$, which is equal to $M+\varepsilon $, when $i=1,\dots,n$, and to $M-\varepsilon $, when $i=n+1,\dots,2n$. By the definition of $M$, only $p_i$ with $i=n+1,\dots,2n$, can belong to $F(C)$.

Exactly $2^n$ cubes of the tiling $\QQ(\varepsilon )$ contain $a$. Each of these is determined by $a$ together with $n$ lattice points chosen among those nearest to $a$, i.e. among the $p_i$ $(i=1,\dots,2n)$. From the previous paragraph, we see that there is only one possible choice, namely $B:=\{a,a-\varepsilon e_1,\dots,a-\varepsilon e_n\}$. Since, by definition, $a$ must belong to some cube of $C$, we conclude that $Q(a)$, the cube defined by $B$, is the only one that contains $a$ and belongs to $C$, as desired.

It is now straightforward to prove (\ref{e:lat-cubes}) by induction on $m=|C|$. Indeed, the inequality is obvious for $m=1$. For the inductive step, suppose (\ref{e:lat-cubes}) holds for every set of no more than $m$ cubes. For a subset $C$ with $m+1$ elements, find $a$ and $Q(a)$, and remove $Q(a)$. The resulting $C'$ has $m$ cubes, $a\notin F(C')$, and $C=C'\cup Q(a)$. By induction, $|F(C')|\geq |C'|$, and $|F(C)|\geq |F(C)'|+1\geq |C'|+1=|C|$, as desired.
\end{proof}

\begin{prop}\label{p:ineq-T}
Suppose that $X\subseteq \R^n$ is a compact space with $\nabla(X)=0<\Delta(X)$, and $F_k\rightarrow X$ as in Proposition~\ref{t:FapproxR^n}. If  
$\overline{T} _{\nabla(F_k)}(X):=|\QQ(\nabla(F_k),X)|$, then
\begin{equation}\label{e:ineq1}
(1/2)\cdot\overline{T}_{\nabla(F_k)}(X)\leq T_{\nabla(F_k)}(F_k)\leq 2^{n-1}\cdot\overline{T}_{\nabla(F_k)}(X).
\end{equation}
Hence also:
\begin{equation}\label{e:ineq2}
2^{1-n}\cdot T_{\nabla(F_k)}(F_k)\leq \overline{T}_{{\nabla(F_k)}}(X)\leq 2\cdot T_{\nabla(F_k)}(F_k).
\end{equation}
\end{prop}

\begin{proof}
Apply Lemma~\ref{l:vert-cubos} to $C= \QQ(\nabla(F_k),X)$, to obtain $|F_k|\geq \overline{T}_{\nabla(F_k)}(X)$. Notice that any subset $U$ of $F_k$ with three or more elements has diameter $>\nabla(F_k)$, because all points of $F_k$ belong to the lattice. As a consequence, every $U$ in a 2-covering $\U$ of $F_k$, with $\DE(\U)=\nabla(F_k)$, contains exactly two elements. Hence, for any such $\U$, $2\cdot|\U|\geq |F_k|$. It follows that $2\cdot T_{\nabla(F_k)}(F_k)\geq |F_k|\geq \overline{T}_{\nabla(F_k)}(X)$, which proves the first inequality of (\ref{e:ineq1}).

The second inequality of (\ref{e:ineq1}) follows from the fact that every cube has a 2-covering by $2^{n-1}$ sets. To see this assume, without loss of generality, that we have a cube with base contained in the hyperplane $x_n=0$, and top contained in $x_n={\nabla(F_k)}$. Then each element of the covering has the form $\{(x',0),(x',{\nabla(F_k)})\}$, where $(x',0)$ runs over the $2^{n-1}$ corners of the base of the cube. Hence, $F_k$ can be covered by at most $2^{n-1}\cdot \overline{T}_{\nabla(F_k)}(X)$ sets of two elements each, as was to be proved. Finally, note that (\ref{e:ineq2}) follows immediately from (\ref{e:ineq1}), as desired.
\end{proof}

\subsection{The Convergence Theorems}\label{ss:convThms}

In this section we complete the proofs of the convergence theorems.

\begin{thm}\label{t:ConvNet}
Let $X\subseteq \R^n$ be compact, with $\nabla(X)=0<\Delta(X)$, and $F_k\rightarrow X$ as in Proposition~\ref{t:FapproxR^n}. Then
\begin{equation} \label{e:ConvB}
\lim_{k\rightarrow \infty} \emph{dim} \m(F_k)=\emph{dim} _\boxx(X).
\end{equation}
\end{thm}

\begin{proof}
By Proposition~\ref{p:limDIMfB}:
\[\lim_{k\rightarrow \infty}\mbox{dim} \m(F_k)=\lim_{k\rightarrow \infty} \frac{\ln T_{\nabla(F_k)}(F_k)}{-\ln \nabla(F_k)}\cdot\]
On the other hand, Falconer~\cite[p.~44-45]{Falconer} shows that when $\nabla(F_{k+1})\geq c\cdot\nabla(F_k)$, the Box-counting dimension is given by:
\[ \mbox{dim} _\boxx(X)= \lim_{k\rightarrow \infty} \frac{\ln \overline{T}_{\nabla(F_k)}(X)}{-\ln \nabla(F_k)}\cdot\]
Hence, to prove~(\ref{e:ConvB}), it suffices to show:
\begin{equation}\label{e:ConvLim}\nonumber
\lim_{k\rightarrow \infty} \frac{\ln T_{\nabla(F_k)}(F_k)}{-\ln \nabla(F_k)} = \lim_{k\rightarrow \infty} \frac{\ln \overline{T}_{\nabla(F_k)}(X)}{-\ln \nabla(F_k)}\cdot
\end{equation}
Now, suppose that the left-hand side of (\ref{e:ConvLim}) exists and equals $s_0$. Then the inequalities

\[
\frac{\ln((1/2^{n-1})\cdot T_{\nabla(F_k)}(F_k))}{-\ln \nabla(F_k)}\leq \frac{\ln \overline{T}_{{\nabla(F_k)}}(X)}{-\ln \nabla(F_k)}\leq \frac{\ln(2\cdot T_{\nabla(F_k)}(F_k))}{-\ln \nabla(F_k)},
\]
obtained from (\ref{e:ineq2}), show that also $\lim_{k\rightarrow \infty} \frac{\ln \overline{T}_{{\nabla(F_k)}}(X)}{-\ln \nabla(F_k)}$ exists and equals $s_0$, as desired. If we assume instead that it is the right-hand side of (\ref{e:ConvLim}) that exists and equals, say, $s_0$, we proceed in the same way, starting this time from (\ref{e:ineq1}). The proof is complete.
\end{proof}

\begin{thm}\label{t:ConvNetIFS}
Let $X\subseteq \R^n$ be compact, with $\nabla(X)=0<\Delta(X)$, and $F_k\rightarrow X$ as in Proposition~\ref{t:FapproxR^n}. Suppose, moreover, that $X$ is the attractor of an iterated function system (IFS). Then
\begin{equation}\label{e:ConvH}
\lim_{k\rightarrow \infty} \emph{dim} \hh(F_k)=\emph{dim}_\hau (X).
\end{equation}
\end{thm}

\begin{proof}
When $X$ is the attractor on an IFS, we have $\mbox{dim} B(X)=\mbox{dim}_\hau (X)$ (see~\cite[p.~132]{Falconer}). By Proposition~\ref{p:dimLocUnif}, $\mbox{dim} \hh(F_k)=\mbox{dim} \m(F_k)$. Hence, (\ref{e:ConvH}) follows from Theorem \ref{t:ConvNet}.
\end{proof}


\end{document}